\pdfoutput=1
 \documentclass[12pt]{article}
 \usepackage{xcolor}
\usepackage{amsfonts,amsmath}
\usepackage[mathscr]{eucal}
\usepackage{amssymb}
\usepackage{bm}
\usepackage{bbold}
\usepackage{amsthm}
\usepackage{graphicx,subfigure}
\usepackage{dsfont}

\usepackage{ytableau}

\usepackage{overpic}

\usepackage{hyperref}

\usepackage[titletoc,toc,title]{appendix}

\newcommand{\appendixnumberline}[1]{Appendix #1.\space}

\let\oldappendix\appendix
\makeatletter
\renewcommand{\appendix}{%
  \addtocontents{toc}{\let\protect\numberline\protect\appendixnumberline}%
  \renewcommand{\@seccntformat}[1]{%\large
  \bfseries Appendix \Alph{section}. }%
  \oldappendix
}
\makeatother

\theoremstyle{plain}
\newtheorem{thm}{Theorem}

\newtheorem*{lem}{Lemma}

\definecolor{navycol}{RGB}{100,150,160}
   \definecolor{pinkcol}{RGB}{242,55,55}
   \definecolor{bluecol}{RGB}{50,205,50}
   \definecolor{bluecol}{RGB}{30,144,255}
   \definecolor{yellowcol}{RGB}{255,252,134}
   \definecolor{lbluecol}{RGB}{214,252,168}
   
   \usepackage{empheq}

\def\theequation{\arabic{section}.\arabic{equation}}
\newcommand{\be}{\begin{eqnarray}}

\newcommand{\ee}{\end{eqnarray}}
\newcommand{\nn}{\nonumber \\}
\newcommand{\lb}{\label}
\newcommand{\p}[1]{(\ref{#1})}

\newcommand{\ga}{\lower.7ex\hbox{$
\;\stackrel{\textstyle>}{\sim}\;$}}
\newcommand{\vecg}[1]{\mbox{\boldmath $#1$}}
\newcommand{\la}{\lower.7ex\hbox{$
\;\stackrel{\textstyle<}{\sim}\;$}}
\newcommand{\CPn}{\mathbb{CP}^{n-1}}
\newcommand{\CP}{\mathbb{CP}}
\newcommand{\dd}{\partial}

\usepackage{setspace}
\usepackage{tcolorbox}
\tcbuselibrary{theorems}

\newtcbtheorem
  []
  {theorem}
  {Theorem}
  {
    colback=bluecol!20!,
    colframe=blue!50!,
    fonttitle=\sffamily,
  }
  {def}

  \newtcbtheorem
  []
  {lemma}
  {Lemma}
  {
    colback=yellowcol!30!,
    colframe=pinkcol!60!,
    fonttitle=\sffamily,
  }
  {def}

\let\OLDthebibliography\thebibliography
\renewcommand\thebibliography[1]{
  \OLDthebibliography{#1}
  \setlength{\parskip}{4pt}
  \setlength{\itemsep}{0pt plus 0.3ex}
}

\begin{document}

\begin{titlepage}

\vspace*{0.2cm}

\renewcommand{\thefootnote}{\star}
\begin{center}

{\LARGE\bf  Monopole harmonics on $\mathbb{CP}^{n-1}$}\\

\vspace{0.5cm}

\vspace{1.5cm}
\renewcommand{\thefootnote}{$\star$}

\quad {\large\bf Dmitri~Bykov} 
 \vspace{0.5cm}

{\it 
 Steklov Mathematical Institute of Russian Academy of Sciences,\\
Gubkina str. 8, 119991 Moscow, Russia

and 

 Institute for Theoretical and Mathematical Physics,\\
Lomonosov Moscow State University, \\119991 Moscow, Russia}
\\

\vspace{0.1cm}

{\tt bykov@mi-ras.ru, dmitri.v.bykov@gmail.com}\\

\vspace{2mm}

\centerline{and}

\vspace{2mm}

\quad {\large\bf Andrei~Smilga} 
 \vspace{0.5cm}

{\it SUBATECH, Universit\'e de Nantes,}\\
{\it 4 rue Alfred Kastler, BP 20722, Nantes 44307, France;}\\
\vspace{0.1cm}

{\tt smilga@subatech.in2p3.fr}\\

\end{center}
\vspace{0.2cm} \vskip 0.6truecm \nopagebreak

\begin{abstract}
\noindent  
We find the spectra and eigenfunctions of both ordinary and supersymmetric quantum-mechanical   models describing the motion of a charged particle over the $\mathbb{CP}^{n-1}$ manifold in the presence of a background monopole-like gauge field.  The states form  degenerate $SU(n)$ multiplets and their wave functions acquire a very simple form being expressed via homogeneous coordinates. Their relationship to  multidimensional orthogonal polynomials of a special kind is discussed. By the well-known isomorphism between the twisted Dolbeault and Dirac complexes, our construction also gives the eigenfunctions and eigenvalues of the Dirac operator on complex projective spaces in a monopole background.
 \end{abstract}

\vspace{1cm}
\bigskip

\newpage

\end{titlepage}

\tableofcontents

\setcounter{footnote}{0}

\setcounter{equation}0
\section{The Dirac monopole and its $\CPn$ generalization} 
Back in 1931, a famous paper by Dirac appeared where he introduced the notion of magnetic monopole \cite{Dirac}. In the same year, Tamm published a paper where he solved the Schr\"odinger equation for the  Hamiltonian
   $
\,\hat H \ =\  {(\hat {\vecg{p}} - e\vecg{A})^2}/(2m)
  $
describing the motion of a scalar  particle of electric charge $e$ and mass $m$ in the field of a monopole \cite{Tamm}.  

To solve the mathematical problem addressed in this paper, we do not need to keep the physical constants. We will set $m = 1$ and $e = -1$ (reflecting the negative sign of the electron charge) in the following. This gives 
\be
\lb{H-bez-konstant} 
\hat H  \ = \   \frac {(\hat {\vecg{p}} + \vecg{A})^2}{2}\,.
\ee
The field density $\vecg{B}$ of the magnetic monopole is spherically symmetric and the wave functions are expressed in terms of spherical harmonics of a special kind, the {\it monopole harmonics}.
The vector potential $\vecg{A}$ derived by Dirac included a singularity (the {\it Dirac string}). But it was later realized that this field, when projected on $S^2$ (the radial motion decouples), can be described as a topologically nontrivial fiber bundle over the 2-sphere \cite{Wu-Yang}. To this end, one has to introduce two charts such that the fields $A^{(1)}$ and  $A^{(2)}$ in these charts are not singular and are related by a gauge transformation 
\be
\lb{gauge12}
A^{(1)}_\mu = A_\mu^{(2)} + \partial_\mu \chi
  \ee
   in the overlapping region. A necessary requirement is that the element $\omega = e^{i\chi}\in U(1)$  is  uniquely defined. 

For $S^2$, one can choose  the whole sphere except one point (its north pole) as chart 1 and a small neightbourhood of the pole as chart 2. Chart 1 can be parametrized by a complex coordinate\footnote{We mostly use the conventions of Ref. \cite{book}. For complex tensors, we will distinguish between four types of indices: covariant and contravariant, holomorphic and antiholomorphic. The indices are  lowered  and  raised with the metric $ h_{j \bar k}$ and its inverse $h^{\bar k j}$ so that, when an index changes a position, it acquires or loses the bar. But to make the formulas more readable, we will not put the bars over the indices for the complex coordinates $\bar z^j$, derivatives $\bar \partial_j$, gauge fields $\bar A_j$ and momenta
$\bar \pi_j$.}  
$z = (x + iy)/\sqrt{2}$ describing the stereographic projection onto the plane tangent to the south pole. Then the metric of the ``round" sphere reads
 \be
ds^2 \ =\ \frac { 2dz d\bar z} {(1 + z \bar z)^2} \,.
 \ee
It is regular at finite $z$. In the limit $z \to \infty$, we arrive at the north pole of $S^2$. 

Consider the gauge fields $A^{(2)} \approx 0$ and 
\be
\lb{A-CP1}
&&A^{(1)}_z = \frac 1{\sqrt{2}} (A_x - iA_y) = -\frac {iq}2 \frac \partial {\partial z} \ln(1+ z \bar z) \ = \ -\frac {iq \bar z}{2(1+ z \bar z)}\,, \nn
&&A^{(1)}_{\bar z} = \overline{ A^{(1)}_z}  = \frac 1{\sqrt{2}} (A_x + iA_y) \ =\  \frac {iq z}{2(1+ z \bar z)} \,.
  \ee
Being projected onto a sphere centered at the origin, the Hamiltonian \p{H-bez-konstant} acquires the form
\be
\lb{nePauli-proj}
\hat H \ =\ -\frac{(1 + z \bar z)^2}2 \left[(\partial_x + iA_x)^2 + (\partial_y + iA_y)^2  \right] \ =\ - \frac{(1 + z \bar z)^2}2 \left\{\frac \partial {\partial z} +iA_z,  \frac \partial {\partial \bar z} + iA_{\bar z}\right\}_+ .
\ee
For large $|z|$, 
\be
A_x^{(1)} \approx -\frac {qy}{x^2 + y^2}, \qquad A_y^{(1)} \approx \frac {qx}{x^2 + y^2}\,.
\ee
It is singular at the north pole (the circulation $\oint A_\mu dx^\mu$ along the small contour around the north pole does not vanish).
It is related to $A^{(2)} \approx 0$ by a gauge transformation \p{gauge12} with $\chi(x,y) = q \arctan(y/x)$. To ensure that $e^{i\chi(x,y)}$ is uniquely defined, $q$ should be integer.\footnote{One can also consider the spectral problem on $\mathbb{CP}^1$ with one removed point for the Hamiltonian   \p{nePauli-proj} with a fractional magnetic charge \cite{flux}. But  such Hamiltonians cannot be supersymmetrized,  while the main point of interest of this paper are supersymmetric $\mathbb{CP}^{n-1}$  models. So we stick to models with integer topological charges.}   The field \p{A-CP1} coincides with the projection of  Dirac's 3-dimensional field on a sphere centered at the origin. Now $q$ is  nothing but the  magnetic charge of the monopole.

Up to the factor of $q$, the field \p{A-CP1} coincides with the connection in the tautological line bundle 
 \be
A^{\rm taut}_z \ =\ \frac i4  \frac \partial {\partial z} \ln \det h \ =\ - \frac i2   \frac \partial {\partial z} \ln (1 + z \bar z)\,.
\ee

The reader could as well consult Ref.~\cite{Dunne} for a detailed discussion of eigenfunctions on the sphere written in terms of inhomogeneous variables. In this paper, we will mostly use \emph{homogeneous} variables. As we will argue, this simplifies the task considerably.

\vspace{1mm}
\subsection{$\CPn$ monopoles}

When viewed as a complex manifold, the sphere $S^2$ is known as the complex projective space~$\mathbb{CP}^1$. It has natural higher-dimensional generalizations, denoted $\mathbb{CP}^{n-1}$, which are 
defined as the sets of complex $n$-tuples $(w^0, ... , w^{n-1})$  identified under the multiplication by a nonzero complex number $\lambda$, i.e. $(w^0, ... , w^{n-1}) \equiv (\lambda w^0, ... , \lambda w^{n-1})$. They admit topologically nontrivial gauge fields also for $n > 2$.

We choose the 
 $SU(n)$-invariant  metric on $\mathbb{CP}^{n-1}$ (the Fubini-Study metric):
\be
  \lb{metr-homogen}
  ds^2 \ =
   {1\over \mathcal{X}} \left(\sum_{\alpha}^n\, d  w^\alpha d \bar w^\alpha -{1\over \mathcal{X}}\,\big|\sum_{\alpha}^n \bar{w}^\alpha dw^\alpha\big|^2\right) \,,
  \\
  \textrm{where}\quad\quad \mathcal{X}:=\sum_{\alpha=1}^n\,w^\alpha \bar{w}^\alpha\,.
   \ee
The latter shorthand notation will be widely used throughout the paper.

The complex projective space $\mathbb{CP}^{n-1}$ can be covered by $n$ charts, each with topology $\mathbb{C}^{n-1}$.
One of these charts excludes points with $w^0 = 0$. The complex coordinates uniquely describing points in this chart can be chosen as $z^j = w^j/w^0, j = 1, \ldots , n-1$. Then the metric \p{metr-homogen} reduces to $ds^2 = 
2 h_{j \bar k} dz^j d\bar z^{k}$ with
 \be
 \lb{Fubini-Study}
h_{j  \bar k}  \ =\ \partial_j  \bar \partial_{k} \, \ln(1 + z^l \bar z^l) \ =\ 
\frac 1{1 + z^l \bar z^l} \left( \delta_{j\bar k} - \frac {\bar z^j z^k}{1 + z^l \bar z^l}
\right) .
 \ee
Evidently, this metric is K\"ahler.

In analogy with \p{A-CP1},  consider   the  gauge field
\be
\lb{A-CP}
A_j \ = \ -\frac {iq}2 \frac \partial {\partial z^j} \ln(1+ z^l \bar z^l) \ = \ -\frac {iq \bar z^j}{2(1+ z^l \bar z^l)},  \qquad A_{\bar j} \ =\ \overline{A_j}
  \ee
in the chosen chart. One can perform a similar analysis\footnote{See e.g. section 13.3 in the book \cite{book}. Basically, one should require  for the flux of the gauge field through a sphere $\CP^1\subset \CPn$ to be integer:
$${1\over 2\pi}\int\limits_{\CP^1} dA \ =\ q\in \mathbb{Z}\,.
$$} to what we have just done for $\mathbb{CP}^1$ and conclude that $q$ must be integer. We will call the field \p{A-CP} a {\it $\mathbb{CP}$ monopole} of charge $q$. 

The spectrum and eigenfunctions of the Hamiltonian 
\be
\hat H \ =\ - h^{\bar kj} \partial_{\bar k} \partial_j
 \ee
on $\mathbb{CP}^{n-1}$, with the tensor $h^{\bar k j}$ being the inverse Fubini-Study metric\footnote{One can check that the operator 
$h^{\bar k j} \partial_j  \bar \partial_k $  coincides  up to the factor of $2$ with the Laplace-Beltrami operator on $\mathbb{CP}^{n-1}$,
\be
2 h^{\bar k j}  \partial_j  \bar \partial_k = \triangle_{LB} = \frac 1{\det(h)} \partial_j  \det(h)\, h^{\bar k j} \bar\partial_k + 
\frac 1{\det(h)} \bar\partial_k  \det(h)\, h^{\bar k j} \partial_j\,,\quad\quad \det(h) \ =\ \frac 1{(1 + z^l \bar z^l)^n}\,.
\ee}
 [this Hamiltonian  is a multidimensional analog of \p{nePauli-proj} in the absence of the gauge field] were found in~\cite[Chapter 3, C. III
]{BGM}, \cite{Wallach}. In a more general case, when the $\mathbb{CP}$ monopoles \p{A-CP} are present, the problem was solved  by Kuwabara~\cite{Kuwabara}.
In the next section, we will rederive his result using a different method\footnote{There is also a general representation-theoretic method based on the identification of the Laplacian (both with and without monopoles) with a Casimir operator acting on a certain (induced) representation of the symmetry group, cf.~\cite{Stone, Landsman1, Landsman2} and references therein. As compared to those methods, the approach we propose in the present paper is of a more algebraic nature and allows to solve the problem in very explicit terms.}.  The supersymmetric case will be considered in Sect. 3.

\section{Ordinary $\mathbb{CP}^{n-1}$}\label{Ordsec}

\setcounter{equation}0

Our approach 
has its roots in two-dimensional sigma models with target spaces such as $\CPn$, which have recently been reformulated in a more algebraic form by one of the authors~\cite{Bykov1, Bykov2, Bykov3}. Physically, this relation is an equivalence between sigma models and the  
Gross-Neveu models, i.e. models with quartic interactions. Quantum mechanical models considered in the present paper can be seen as dimensional reductions of those sigma models.

\subsection{Hamiltonian formalism}
We start with a mechanical problem of a particle on $\mathbb{CP}^{n-1}$ equipped with the metric~(\ref{metr-homogen}). The relevant Lagrangian is
\be
\lb{L-q0}
L\ =\ \frac{\dot{w}^\alpha \Pi_{\alpha \bar \beta} \dot{\bar{w}}^\beta}{{\cal X}}\,,
 \ee
where ${\cal X} =  w^\alpha \bar w^\alpha$ and  $\Pi_{\alpha \bar \beta} =\delta_{\alpha\bar \beta}  - {\bar{w^\alpha}w^\beta}/{\cal X}$
is the projector on the directions  orthogonal to vector $w^\alpha$,
 \be
\lb{projector}
w^\alpha \Pi_{\alpha\bar \beta} \ =\ \bar w^{\beta}   \Pi_{\alpha\bar\beta} \ =\ 0 \,.
 \ee
 In the chart excluding the points with $w^0 = 0$, this Lagrangian can be rewritten in terms of the variables $z^j = w^j/w^0$ and the metric \p{Fubini-Study}: 
\be 
\lb{L-CP-via-z}
 L \ = \ h_{j\bar k} \dot z^j \dot {\bar z}^{k}\,.
 \ee
But the form \p{L-q0} is more convenient for our purposes. The Lagrangian \p{L-q0} is invariant under the gauge transformations
\be
\lb{gauge-lam}
w^\alpha  \ \to \ \lambda(t) w^\alpha, \quad\quad \lambda(t) \in \mathbb{C}^{\ast}\,.
 \ee

These gauge transformations correspond to the $\mathbb{C}^{\ast}$ quotient entering the definition of the projective space. 

The canonical momenta are
 \be
\lb{pi0}
\pi_\alpha \ =\ \frac{\Pi_{\alpha \bar\beta }\dot {\bar w}^\beta }{{\cal X}}, \qquad \bar \pi_{\alpha} \ =\ \frac{\Pi_{\beta\bar\alpha}\dot {w}^\beta}{{\cal X}}\,.
 \ee
As a consequence of the gauge symmetry \p{gauge-lam}, they obey the constraints
\be
\lb{cons-class}
w^\alpha \pi_\alpha \ = \ \bar w^\alpha \bar\pi_\alpha \ =\ 0 \,.
 \ee
The canonical Hamiltonian is 
\be
\lb{H-class}
H \ = \  {\cal X}  \, \pi_\beta  \bar \pi_{\beta}\,.
 \ee
 
There is an alternative way to see how the constraints~\p{cons-class} emerge. To this end, one can introduce a couple of extra nondynamic variables $A, \bar A$ and express the Lagrangian in the form 
\be
\lb{L-q1}
L' \ =\ \frac{(\dot w^\alpha - A w^\alpha)(\dot {\bar w}^\alpha - \bar A \bar w^{\alpha})}{w^\beta \bar w^\beta} \,.
\ee
The Lagrangian \p{L-q0} is then restored after excluding $A$ and $\bar A$. The constraints 
\p{cons-class} appear from the equations of motion  ${\partial L' \over\partial A}
 = {\partial L' \over \partial \bar A }= 0$.
 
Formulations with auxiliary gauge fields, such as~(\ref{L-q1}), are well known in the context of two-dimensional sigma models
under the name of gauged linear sigma models (GLSM). The GLSM approach to the $\CPn$-system
dates back to~\cite{Cremmer, Vecchia1, Vecchia2}. Those authors, as well as most of their successors, used a version of GLSM where one imposes the constraint $w^\alpha {\bar w}^\alpha=1$, which breaks $\mathbb{C}^\ast$ down to $U(1)$, so that the auxiliary gauge field becomes a $U(1)$ gauge field. 
We prefer not to introduce any constraints and work with the $\mathbb{C}^\ast$ quotient.

\subsection{Quantization} 
 
To quantize the theory, we define the operators   $\hat \pi_\beta= -i \partial/ \partial w^\beta$, $\hat{\bar \pi}_{\beta} = -i \partial / \partial \bar w^\beta $. The standard commutation relations hold:
\be
[\hat \pi_\beta, w^\alpha] \ = \ -i\,\delta^\alpha_\beta\,,\quad\quad 
[ \hat {\bar\pi}_\beta, \bar{w^\alpha}]\ =\ -i\,\delta^\alpha_\beta\,.
\ee
To write the quantum version of \p{H-class}, we should prescribe a particular way of ordering of the operators in $\hat H$. We make the following choice: 
\begin{empheq}[box = \fcolorbox{black}{white}]{align}
\lb{Ham-w}
\quad\hat H \ =\ - {\cal X}  \, \frac {\partial^2}{\partial w^\beta \partial \bar w^\beta} \equiv   - {\cal X} \triangle\,,
\end{empheq}
while the quantum constraints are chosen as 
\be
\lb{cons-bos-q0}
w^\alpha \frac \partial {\partial w^\alpha} \Psi \ =\ \bar w^\alpha \frac \partial {\partial \bar w^\alpha} \Psi  \ =\ 0\,,
\ee
Thus, we are going to solve the spectral problem $\hat H \Psi = E\Psi$  in the Hilbert space
  \be
  \mathcal{H}_0:={\cal L}_2\left(\CP^{n-1}\right)
  \ee
involving  the constraints \p{cons-bos-q0}
 and equipped with the norm\footnote{One could take any rapidly falling function of  ${\cal X}$ as the measure, but $e^{ - {\cal X}}$ is the most natural choice.}
\be
\lb{norm-w}
\langle \Psi|\Psi \rangle \ =\ \int \prod_{\alpha = 0}^{n-1} dw^\alpha d  \bar w^\alpha \exp \{  - {\cal X}\} \, |\Psi|^2 .
 \ee
Note that the Hamiltonian \p{Ham-w} is Hermitian in this Hilbert space: when we flip the derivative $\partial/\partial w^\alpha$ in the matrix element
$$ \int \prod_{\alpha = 0}^{n-1} dw^\alpha d  \bar w^\alpha \, \bar \Psi_1(w, \bar w) F({\cal X}) \frac {\partial^2}{\partial w^\alpha \partial \bar w^\alpha} \Psi_2(w, \bar w) \,,$$ 
it acts only on $\bar \Psi_1$. The contribution arising from differentiating $ F({\cal X})$ vanishes due to  \p{cons-bos-q0}.

\vspace{1mm}

In the unconstrained formulation, the spectral problem is
\be
-h^{\bar k j}  \partial_j   \bar \partial_k \Psi \ =\ -(1 + z^l \bar z^l) (\bar z^k z^j + \delta^{\bar k j} )  \, \partial_j  \bar \partial_k  \Psi \ =\ E\Psi\,,
\ee
and the norm \p{norm-w} transforms into the norm 
\be
\lb{norm-z}
\langle \Psi|\Psi \rangle \ =\ \int \prod_{j=1}^{n-1} dz^j d  \bar z^j  \, \det(h) \, |\Psi|^2 
 \ee
in the unconstrained Hilbert space.

\subsection{Eigenfunctions and eigenvalues}

\subsubsection{Without monopoles}
We now proceed to solve the spectral problem formulated above, sticking to the description \p{Ham-w}, \p{cons-bos-q0} with unresolved constraints.  We will be looking for eigenfunctions of  \p{Ham-w} satisfying \p{cons-bos-q0} of the following general form 
\be
\lb{Psi-Ans-q0}
\quad\Psi(w, \bar w) \ =\ \frac 1{{\cal X}^L} A_{\alpha_1 \ldots \alpha_L|\bar\beta_1 \ldots \bar\beta_L} w^{\alpha_1} \cdots w^{\alpha_L} \, 
 \bar w^{\beta_1} \cdots \bar w^{\beta_L}\quad
\ee
with an integer $L \geq 0$. 

\vspace{0.3cm}

\begin{thm} \lb{trace}
\vspace{0.3cm}
The function \p{Psi-Ans-q0} is an eigenfunction of the Hamiltonian \p{Ham-w} iff 
 the tensor $A$ is traceless,
\be
A_{\gamma \ldots \alpha_L|\bar\gamma \ldots \bar\beta_L} \ =\ 0 \,.
 \ee
 The eigenvalue is 
 \be\label{energbos}
 E(n, L, q=0) \ =\ L(L+ n-1)\,.
 \ee
\vspace{-0.7cm}
\end{thm}

\begin{proof}
First, if $A$ is traceless, it can be easily checked that $\hat H \Psi = L(L+n-1) \Psi$, 
so that~\p{Psi-Ans-q0} is an eigenfunction 
with the eigenvalue~(\ref{energbos}).

If $A$ is not traceless, one can subtract its trace parts, recasting the wave function in the form
\be
\Psi(w, \bar w) \ =\ \frac 1{{\cal X}^L} \widehat{A}_{\alpha_1 \ldots \alpha_L}^{\bar\beta_1 \ldots \bar\beta_L} w^{\alpha_1} \cdots w^{\alpha_L} \, 
 \bar w^{\beta_1} \cdots \bar w^{\beta_L} \nn
 +\ \frac 1{{\cal X}^{L-1}} \widehat{A}_{\alpha_1 \ldots \alpha_{L-1}}^{\bar\beta_1 \ldots \bar\beta_{L-1}} w^{\alpha_1} \cdots w^{\alpha_{L-1}} \, 
 \bar w^{\beta_1} \cdots \bar w^{\beta_{L-1}}+\ldots \,,
\ee
where the tensors  $\widehat{A}$ are traceless. The Hamiltonian acts on each of the terms via multiplication by the eigenvalue~(\ref{energbos}) with different values of $L$. Thus, $\Psi$ is an eigenfunction if only one of the terms is present. This implies that $A$ is traceless.
\end{proof}

The ground state of $\hat H$ is $\Psi = const$. The corresponding eigenvalue is $E=0$. At the next level, we have the functions
\be
\Psi = \frac {A_{\alpha|\bar\beta} \,w^\alpha \bar w^{\beta} }{\cal X}
\ee
with $A_{\alpha|\bar\alpha} =0$. There are $n^2-1$ such independent functions. All of them have the energy $E= E(n,L=1,q=0) = n$. 
These $n^2-1$ degenerate states belong to the adjoint representation of $SU(n)$. Such a degeneracy should have been expected, because the Hamiltonian  \p{Ham-w} is obviously invariant under unitary transformations.

\begin{figure}
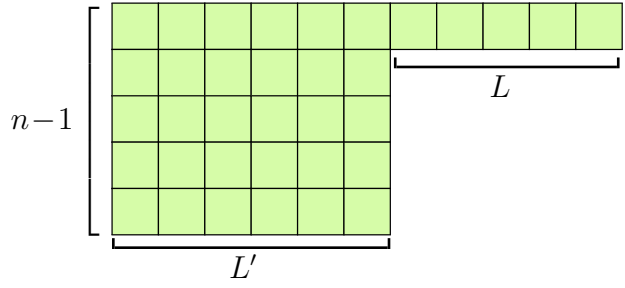

\centering
\begin{Overpic}[]{
\ytableausetup{boxsize=0.6cm}
\ytableaushort
{}
* [*(lbluecol)]{11,6,6,6,6}
}
\put(-15,20){$n\!-\!1 \left[
\begin{array}{rrr}
\\
\\
\\
\\
\\
\\
\end{array}
\right. $}
\put(57,34){$\underbracket[1pt][0.6ex]{\hspace{3cm}}_{}$}
\put(75,26){$L$}
\put(4,0){$\underbracket[1pt][0.6ex]{\hspace{3.7cm}}_{}$}
\put(26,-8){$L'$}
\end{Overpic}
\vspace{0.5cm}
\caption{Young diagram for the relevant $SU(n)$ representations.}
\label{young1}
\end{figure}

For higher $L$, the representations may be described by the Young diagram shown in Fig.~\ref{young1} (with $L'=L$).  Alternatively, we can describe it by using Dynkin labels. They are simply the differences of the row lengths of the Young diagram. As a result, these representations of $SU(n)$ have only two nonzero Dynkin labels for the leftmost and the rightmost node of the 
Dynkin diagram:
\be
\label{DynkinBosonic}
\quad a_{i = 1,\ldots, n-1} = (L, \underbracket[1pt][0.6ex]{\;0, \;\ldots, \;0,}_{n-3}\; L)\,.\quad 
\ee

The dimensions of such representations\footnote{See Appendix A.} are
 \be
\#(n,L) \ =\ \left( \begin{array}{c} L+n-2\\ n-2 \end{array}  \right)^2 \frac {2L + n-1}{n-1}\,.
\ee
For example, for $n =3$, $\#(3,L) = (L+1)^3$. We have a singlet, an octet, a 27-plet, etc.

Hermiticity of the Hamiltonian means that the inner prodict of two wave functions \p{Psi-Ans-q0} with different $L$'s and hence with different energies vanishes. It is not always
 so for different functions with the same $L$, but an orthogonal basis  can always be chosen.  
\vspace{1mm}

\subsubsection{With monopoles}

To describe a system in the presence of a $\mathbb{CP}$ monopole of charge $q\in \mathbb{Z}$, we replace the Lagrangian \p{L-q0} by 
\be
\tilde{L}\ =\ L - \dot w^\alpha A_\alpha - \dot{\bar  w}^\alpha \bar A_\alpha \,,
\ee
where 
\be
\lb{CP-monopole-w}
A_\alpha \ =\ -i\frac {q \bar w^\alpha}{2{\cal X}}, \quad \bar A_\alpha \ =\  i\frac {q  w^\alpha}{2{\cal X}}
\ee
is the monopole gauge field \p{A-CP} in homogeneous coordinates. 

The action corresponding to the new Lagrangian is still invariant under the $\mathbb{C}^\ast$ gauge transformations $w^\alpha   \to  \lambda(t) w^\alpha$. { The canonical Hamiltonian 
reads 
\be
H_q \ = \ {\cal X}  (\pi_\alpha +  A_\alpha)  (\bar \pi_\alpha + \bar A_\alpha)\,.
 \ee
For its quantum counterpart, we choose a natural expression
\be
\lb{Bochner}
\hat H_q \ =\ - \frac {{\cal X}}2 \{( \partial_\alpha + i  A_\alpha) , \, (\bar  \partial_\alpha + i  \bar A_\alpha) \}_+  \
 =\ - \frac {{\cal X}}2\left\{\left(\frac \partial{\partial w^\alpha} - \frac {q\bar w^\alpha}{2{\cal X}} \right),\,
\left(\frac \partial{\partial \bar w^\alpha}  + \frac {q w^\alpha}{2{\cal X}} \right) \right\}_+\,,
 \ee
the operator known as {\it Bochner Laplacian}.}

The constraints are modified as follows:
\be
\lb{cons-bos-q}
w^\alpha \frac \partial {\partial w^\alpha} \Psi \ =\  - \frac q2 \Psi, \qquad \bar w^\alpha \frac \partial {\partial \bar w^\alpha} \Psi  \ =\ \frac q2 \Psi.
\ee
This means that, 
 when  the homogeneous coordinates are multiplied by a common phase factor, $w^\alpha\to e^{i \theta} w^\alpha$, the wave function 
transforms as $\Psi\to e^{-i q\theta} \Psi$, implying that it is a section of the $q$th power of the Hopf bundle over $\CPn$. On the space of such sections, we define the same scalar product~(\ref{norm-w}), turning it into the Hilbert space that we denote as
\be
\mathcal{H}_q=\mathcal{L}_2^{(q)}\left(\CPn\right)\,.
\ee

We will be looking for eigenfunctions of \p{Ham-w} satisfying the constraints \p{cons-bos-q} in the form
\begin{empheq}[box = \fcolorbox{black}{white}]{align}
\lb{Psi-Ans-q}
&\quad \Psi(w, \bar w) \ =\ \frac 1{{\cal X}^{{L+L'\over 2}}} A_{\alpha_1 \ldots \alpha_L| \bar\beta_1 \ldots \bar \beta_{L'}} w^{\alpha_1} \cdots w^{\alpha_L} \, 
 \bar w^{\beta_1} \cdots \bar w^{\beta_{L'}}\quad,\\ \nonumber
 &\quad\textrm{where}\quad\quad L'-L=q \,. 
\end{empheq}

As earlier, for $\Psi$ to be an eigenfunction, the tensor $A_{\ldots}$ should be traceless. The eigenfunctions form $SU(n)$ multiplets that are characterized by the sets of Dynkin labels
\be
\quad (L, \underbracket[1pt][0.6ex]{\;0, \;\ldots, \;0,}_{n-3}\; L')\quad
\ee
The dimensions of these representations are 
\be
\lb{dim-q}
\#(n, L, L') \ =\ \left(\begin{array}{c} L+n-2\\ n-2 \end{array}  \right) \left(\begin{array}{c} L' + n-2\\ n-2 \end{array}  \right) \frac {L+L'+ n-1}{n-1}\,.
\ee
The energies are 
\be
\lb{spektr-bos-q}
\quad E(n, L, L') \ =\ L L' + (n-1) \frac {L + L'}2
 \ee

The wave functions \p{Psi-Ans-q0} and the  multiplicities \p{dim-q} have been derived earlier by Kuwabara~\cite{Kuwabara} by a different method: he considered the fiber bundle
$S^{2n-1} \to \mathbb{CP}^{n-1}$ and picked up the spherical harmonic on $S^{2n-1}$ that ``survive" the projection. 

 \vspace{1mm}

It is also instructuve to write out an explicit expression for the  Hamiltonian in inhomogeneous coordinates.
In the chart excluding the points with $w^0 = 0$, a wave function satisfying \p{cons-bos-q} can be presented in the form  
\be
\Psi(w^\alpha, \bar w^\alpha) \ =\  \left( \frac {\bar w^0}{w^0}\right)^{q/2}\Phi(z^j, \bar z^j) \quad {\rm with} \quad  z^j = \ \frac {w^j}{w^0}.
  \ee 
 The wave functions $\Phi(z^j, \bar z^j)$ live in the ${\cal L}_2$ Hilbert space with the norm
\p{norm-z}.

One can derive that
\be
%-
\hat H_q \Psi(w, \bar w) \ =\ \left( \frac {\bar w^0}{w^0}\right)^{q/2} \hat H'_q \Phi(z^j, \bar z^j) \,,
\ee
where
\be
\lb{H-Phi}
\hat H'_q &=&  - \frac 12 h^{\bar k j}  \left\{\frac \partial {\partial z^j} + iA_j,  \frac \partial {\partial \bar z^{\bar k}} + i\bar A_k \right\}_+ \nn 
&=& (1 + z^l \bar z^l) \left[ - (\delta^{jk} +  z^j \bar z^k) \frac{\partial^2}{\partial z^j \partial \bar z^k}  + \frac q2 \left(z^j \frac {\partial}{\partial z^j} - 
\bar z^k \frac {\partial}{\partial \bar z^k} \right) \right] + \frac{q^2  z^l \bar z^l}4\,.
\ee  
This Hamiltonian  represents a multidimensional generalization of \p{nePauli-proj}. 

\subsection{Orthogonal polynomials}

It is well-known that  the eigenfunctions of the Laplacian on the sphere $S^2\simeq \CP^1$ are expressed in terms of the (associated) Legendre polynomials.
Eigenfuctions of the Hamiltonian \p{Ham-w} are also expressed in terms of certain multidimensional orthogonal polynomials.

\subsubsection{$n=2$}

Consider first the simplest $\mathbb{CP}^1$ case. The eigenstates form $SU(2)$ multiplets of dimension $L+ L' +1$. The wave functions are expressed via the {\it Jacobi polynomials} \cite{Tamm,Wu-Yang}, which can also be understood in our approach\footnote{The relation of Jacobi polynomials to the representation theory of $SU(2)$ is nicely explained in~\cite{Gurarie}.}. Consider as an example the case $q=1$. At the level $L=0$ (and hence $L' = 1$), we have two states:
\be
\lb{q1-L0}
\Psi_{L=0}(w^{0,1}, \bar w^{0,1}) \ =\ \frac {\bar w^0}{\sqrt{w^\alpha \bar w^\alpha}} \quad {\rm and} \quad \frac {\bar w^1}{\sqrt{w^\alpha \bar w^\alpha}}\,.
\ee
This gives
\be
\Phi_{L=0}(z, \bar z) \ =\ \frac 1{\sqrt{1 + z \bar z}} \quad {\rm and} \quad  \frac {\bar z}{\sqrt{1 + z \bar z}}\,.
 \ee
At the level $L=1$, a generic function \p{Psi-Ans-q} reads 
\be
\Psi_{L=1}(w^{0,1}, \bar w^{0,1}) \ =\ \frac 1  {{\cal X}^{3/2}} A_{\alpha|\bar\beta_1 \bar\beta_2} \, w^\alpha \bar w^{\beta_1} \bar w^{\beta_2}
\ee
with the tracelessness condition to be imposed on $A_{\alpha|\bar\beta_1 \bar\beta_2}$. This gives in terms of $z, \bar z$ four functions:
\be
\lb{q1-L1}
\Phi_{L=1}(z, \bar z)\, :\ =\  \frac z{(1 + z \bar z)^{3/2}}, \quad  \frac {1 - 2 z \bar z}{(1 + z \bar z)^{3/2}}, \quad \frac {\bar z( 2- z \bar z )}{(1 + z \bar z)^{3/2}},
 \quad \frac {\bar z^2}{(1 + z \bar z)^{3/2}}\,.
\ee
For $L=2$, there are   six functions:
\be
\lb{q1-L2}
\Phi_{L=2}(z, \bar z) \, : \ =  \frac {z^2}{(1 + z \bar z)^{5/2}}, \quad  \frac {z(2 - 3z \bar z)}{(1 + z \bar z)^{5/2}}, \quad  \frac {1 - 6 z \bar z + 3 z^2 \bar z^2}{(1 + z \bar z)^{5/2}},  \nn
\quad  \frac {\bar z(3 - 6 z \bar z + z^2 \bar z^2)}{(1 + z \bar z)^{5/2}}, \quad 
 \frac {\bar z^2(3 - 2 z \bar z)}{(1 + z \bar z)^{5/2}}, \quad  \frac {\bar z^3}{(1 + z \bar z)^{5/2}}\,.
\ee

The functions \p{q1-L0}, \p{q1-L1}, \p{q1-L2} as well as all the functions with still higher $L$  are normalizable and mutually orthogonal with the measure 
 \be
\lb{measure2}
  d\mu \ =\ \frac {dz d\bar z}{(1 + z \bar z)^2}\,. 
 \ee
  
Introducing 
$$e^{i\phi} = \frac z{\sqrt{z \bar z}}  \quad {\rm and} \quad t = \cos \theta  = \frac {1-z\bar z}{1 + z\bar z}\,, $$ this measure acquires a flat form, $d\mu = dt d\phi$ with
$\phi \in [0,2\pi)$ and $t \in [-1,1]$.

 The eigenstates are distinguished by the (integer) eigenvalues of the angular momentum operator $\hat m = -i \dd/\dd \phi$. First let us pick the states  with $m=0$. Their wave functions read $\Phi_L^{m=0}(t) \sim \sqrt{1 + t} \, P_L(t)$, where $P_L(t)$ are polynomials of degree $L$. The functions $\Phi_L^{m=0}(t)$ are normalized in the interval $[-1, 1]$ with the flat measure. It follows that $P_L(t)$ are normalized in this interval with the measure $d\mu = (1+t) dt$ and represent the  Jacobi polynomials\footnote{By definition, the Jacobi polynomials $P_n^{\alpha\beta}(t)$ are orthogonal in the interval  $[-1, 1]$ with the measure $(1-t)^\alpha (1+t)^\beta$. The first three  polynomials read
\be
\lb{Jacobi}
&&P_0^{\alpha \beta} =1, \quad P_1^{\alpha \beta}(t) = \alpha+1 + \frac {\alpha + \beta + 2}2(t-1), \nn
&&P_2^{\alpha \beta}(t)  = \frac {(\alpha+1)(\alpha+2)}2 + \frac {(\alpha+2)(\alpha+\beta+3)}2 (t-1) + \frac {(\alpha+\beta + 3)(\alpha+\beta + 4)}8 (t-1)^2.
\ee
Note the relations \cite{Wu-Yang}
\be
\lb{relations}
&&P_{s+\alpha}^{-\alpha, \beta}(t) \ =\ C(\alpha,\beta,s) (t-1)^\alpha P_s^{\alpha,\beta}(t), \nn
&&P_{s+\beta}^{\alpha, -\beta}(t) \ =\ C'(\alpha,\beta,s) (t+1)^\beta P_s^{\alpha,\beta}(t) \,, 
\ee
which hold for integer $\alpha,\beta$. 
}  $P_L^{0,1}(t)$. 
The wave functions in the $m=1$ sector have the form
\be
\Phi_L^{m=1}(t,\phi) \ \sim \ e^{i\phi} \sqrt{1-t}\,(1+t) \, P_{L-1}^{1,2}(t) \qquad  (L \geq 1). 
\ee
The general formula for any positive $q$ and any $m$ reads [to verify that these expressions coincide with the wave functions in  \p{q1-L0}, \p{q1-L1}, \p{q1-L2}, use the relations \p{relations}]
 \be
\Phi^m_L(t, \phi) \ \sim \  e^{im\phi} (1-t)^{-m/2} (1+t)^{(m+q)/2}\,  P_L^{-m, m+q} (t)
\ee
with $L = L_0,L_0+1, \ldots$, where
  \be
L_0 \ =\ -(m+q) \quad {\rm if} \  m < -q, \quad L_0 = 0   \quad {\rm if} \ -q \leq m \leq 0 \quad {\rm and} \quad L_0 = m  \quad {\rm if} \ m  > 0\,.
 \ee

\subsubsection{$n = 3$}
For higher $n$, the eigenfunctions can also be expressed in terms of orthogonal polynomials, in this case of several variables. Such polynomials were rather extensively studied by mathematicians \cite{Suetin,Dunkl}. We will not come to grips with the description of a generic system but,  to illustrate how the orthogonal polynomials emerge in our problem, we will discuss in detail only the simplest nontrivial case: the motion over $\mathbb{CP^2}$ in the absence of monopoles. The eigenstates of the Hamiltonian \p{Ham-w} with the constraints \p{cons-bos-q0} form the $SU(3)$ multiplets with Dynkin labels $(L,L)$. When $L=0$, the wave function is just a constant. In the simplest nontrivial case $L=1$, we have an octet of degenerate states. Their eigenfunctions can be chosen as
\be
\lb{octet}
&&\Psi_{1,2,4,5,6,7} \ = \ \frac {w^\alpha \bar w^\beta}{\cal X} \quad {\rm with} \quad \alpha \neq \beta, \nn
&&\Psi_3 \ = \ \frac {w^0 \bar w^0 - w^1 \bar w^1}{\cal X} \quad {\rm and} \quad \Psi_8 \ =\ \frac {w^0 \bar w^0 + w^1 \bar w^1 - 2 w^2 \bar w^2}{\cal X} \,.
\ee 
 Let us concentrate on the states $\Psi_3$ and $\Psi_8$ having a zero weight.
Fixing the gauge $w^0 = 1$ and introducing $x = z^1 \bar z^1$, $y = z^2 \bar z^2$, their wave functions acquire the form
 \be
\lb{Psi38}
\Psi_3 \ =\ \frac {1-x}{1+x+y} \qquad {\rm and} \qquad  \Psi_8 \ =\ \frac {1+x - 2y}{1+x+y}\,.
 \ee
The functions \p{octet} and \p{Psi38} are orthogonal with the measure
   \be
\lb{measure3}
  d\mu \ =\ \frac {dz^1 d\bar z^1 dz^2 d\bar z^2}{(1 + z^1 \bar z^1 + z^2 \bar z^2)^3} \ = \ \frac {dx dy d\phi_1 d\phi_2}{(1 + x + y)^3}\,,
 \ee
where $\phi_{1,2}$ are the phases of $z^{1,2}$, of which the functions \p{Psi38} do not depend. Introduce
 \be
t_1\  = \ \frac x{1+x+y}, \qquad  t_2\  = \ \frac y{1+x+y} \,.
 \ee
Then
\be
\lb{measure-t12}
\frac {dx dy}{(1+x+y)^3} = dt_1 dt_2
\ee
and the domain where the variables $t_{1,2}$ change is a triangle (a 2-dimensional simplex\footnote{We note in passing that there is a natural action of the torus $[U(1)]^{n-1}$ on $\CP^{n-1}$. The variables $t_i$ are the moment maps of this action, and the  moment polytope is an $(n-1)$-simplex. This is the geometric reason why these variables are useful.}) 
\be
\lb{simplex}
0 \leq t_1 \leq 1, \qquad 0 \leq t_2 \leq 1 - t_1 \,.
 \ee
Being expressed in terms of $t_{1,2}$, the wave functions \p{Psi38} read
 \be
\lb{Psi38-t}
\Psi_3 \ = \ 1 - 2t_1 - t_2, \qquad \Psi_8\  = \ 1 - 3t_2 \,.
 \ee
They are orthogonal on the simplex \p{simplex} with the flat measure \p{measure-t12}.

Consider now the functions of zero weight in the higher multiplets: the 27-plet with the Dynkin labels $(L=2,L=2)$,  the 64-plet with the labels $(3,3)$, etc. At level $L$, there are $L+1$ such states.
Indeed, a representation $(L,L)$ can be thought of as a symmetrized product of $L$ adjoint representations, and the elements of zero weight in this representation can be thought of as the symmetrized  product of the Cartan subalgebras of $su(3)$. The rank of  $su(3)$ is 2. The number of elements 
at level $L$ is the number of ordered sets $p_1 \leq \cdots \leq p_{L-1}$ where each $p_j$ acquires only two values. There are $L+1$ such sets.

The wave functions of three naturally chosen zero-weight states at level $L=2$ are
\be
\lb{L=2}
&&\Psi_1^{L=2} \ =\ \frac {( w^0 \bar w^0)^2 - 4  w^0 \bar w^0  w^1 \bar w^1  + ( w^1 \bar w^1)^2 }{{\cal X}^2}  \ = \ 
  1 -6t_1 - 2t_2 + 6t_1^2 + t_2^2 + 6 t_1 t_2, \nn
&&\Psi_2^{L=2} \ =\ \frac {( w^0  \bar w^0)^2 - 4  w^0 \bar w^0  w^2 \bar w^2  + ( w^2 \bar w^2)^2 }{{\cal X}^2}  \ = \ 
  1 -2t_1 - 6t_2 + t_1^2 + 6t_2^2 + 6 t_1 t_2, \nn
&&\Psi_3^{L=2} \ =\ \frac {( w^1 \bar w^1)^2 - 4  w^1 \bar w^1 w^2  \bar w^2  + ( w^2 \bar w^2)^2 }{{\cal X}^2}  \ = \ 
t_1^2 + t_2^2 - 4 t_1 t_2.\nn
\ee
They are linearly independent, but not mutually orthogonal: $\langle 1|1 \rangle  =  \langle 2|2 \rangle =  \langle 3|3 \rangle  \ =\ \frac 1{30}$, $\langle 1|2 \rangle  =  \langle 1|3 \rangle =  \langle 2|3 \rangle  \ =\ \frac 1{180}$. 

One can, of course, find their orthogonal linear combinations,  which we will do below. Notice that the functions \p{L=2} are orthogonal to  \p{Psi38} and to $\Psi^{L=0} = 1$. 

\subsubsection{Orthogonal polynomials on the simplex}

The functions $\Psi^{L=2}_{1,2,3}$ are quadratic polynomials in two variables $t_1$ and $t_2$. After orthogonalization, they become orthogonal polynomials on the simplex \p{simplex} with flat measure. The problem of finding a set of orthogonal polynomials on the simplex \p{simplex} with the weight $t_1^\alpha t_2^\beta (1 - t_1 - t_2)^\gamma$ was addessed and solved long time ago by P. Appell \cite{Appell}. They represent a generalisation of the Jacobi polynomials and are expressed through them. In our case, when $\alpha = \beta = \gamma = 0$,  the Appell polynomials have the form\footnote{For a general expression, see Eq.(29) in chapter 10 of the book \cite{Suetin}.} 
\be
A_{p\geq k}(t_1, t_2) \ =\ (1-t_2)\, P^{0,2k+1}_{p-k}(1-2t_2) \, P_k
\left( \frac {2t_1 + t_2 -1}{1-t_2}   \right)  \, ,
 \ee
where $P^{0, 2k+1}_{p-k}$ are the Jacobi polynomials and $P_k$ are the ordinary Legendre polynomials.

The calculation gives
\be
&&A_{00} = 1, \nn
&& A_{10} =  1 - 3t_2, \qquad A_{11} = 2t_1 + t_2 -1, \nn
&&A_{20} = 1-8t_2 + 10 t_2^2, \qquad A_{21} = -1 + 2t_1 + 6t_2 - 5t^2_2 - 10 t_1 t_2, \nn 
&&A_{22} = 2( 1 - 6t_1 - 2t_2  + t_2^2 + 6 t_1^2 + 6 t_1 t_2).
\ee
The second line coincides with  \p{Psi38-t}. The functions $A_{2k}$ represent (mutually orthogonal) linear combinations of the functions \p{L=2}:
\be
A_{20} = \frac {3(\Psi_2^{L=2} + \Psi_3^{L=2}) - \Psi_1^{L=2}}2, \quad A_{21} = \Psi_3^{L=2} - \Psi_2^{L=2}, \quad A_{22} = 2\Psi_1^{L=2}.
 \ee

\vspace{1mm}

A similar analysis can be carried out for higher $n$. Suppose one still has $q=0$, and consider only the wave functions of zero weight. In terms of $z^j, \bar z^j$, the wave functions  \p{Psi-Ans-q0} with different $L$'s are orthogonal to one another with the measure 
\be
\lb{meas-z}
d\mu \ =\ \frac {\prod_{j=1}^{n-1} dz^j d\bar z^j}{(1 + z^k \bar z^k)^n} \ \sim \ \frac {\prod_{j=1}^{n-1} dx_j }{\left( 1 + \sum_{j=1}^{n-1} x_j \right)^n}\,.
 \ee
The functions of the same $L$ can be orthogonalized. Introduce the variables
 \be t_j \ =\ \frac {x_j}{ 1 + \sum_{j=1}^{n-1} x_j } \,. 
\ee

\begin{lem}
\vspace{0.3cm}
The measure \p{meas-z} corresponds to a flat measure 
\be
\lb{flat}
d\mu = \prod\limits_{j=1}^{n-1} dt_j
\ee
on the simplex
\be
\lb{simplex-n}
0 \leq t_1 \leq 1, \ 0 \leq t_2  \leq 1 - t_1, \ \ldots \ ,0 \leq t_{n-1} \leq 1 - \sum_{j=1}^{n-2} t_j \,.
 \ee
 \vspace{-0.3cm}
\end{lem}
\begin{proof} 
The Jacobian reads  (here $R=1+\sum_{j=1}^{n-1} x_j$):
\be 
\lb{Jacobian} 
J \ =\ \frac {D(t_1, \ldots, t_{n-1})}{D(x_1, \ldots x_{n-1})} \ =\ \left| \begin{array}{cccc} \frac 1R - \frac {x_1}{R^2}& \ldots & \ldots & -   \frac {x_1}{R^2}\\
- \frac {x_2}{R^2} & \frac 1R - \frac {x_2}{R^2}& \ldots& - \frac {x_2}{R^2} \\
\ldots& \ldots & \ldots & \ldots \\
- \frac {x_{n-1}}{R^2} &  \ldots & \ldots & \frac 1R - \frac {x_{n-1}}{R^2} 
\end{array}  \right|.
\ee
If  the terms  $1/R$ were crossed out, one would obtain the matrix of rank 1. The determinant and all the minors of such a matrix vanish. This means that the only remaining contributions in the determinant  \p{Jacobian} are
 $$  J \ =\ \frac 1 {R^{n-1}} - \sum_j^{n-1} \frac {x_j}{R^n} \ =\ \frac 1 {R^n}\,, $$
from which \p{flat} follows.
\end{proof}
As in the $n=3$ case, the wave functions at level $L$  represent polynomials of degree $L$ in $n-1$ variables $t_j$, which are orthogonal on the simplex \p{simplex-n} with the flat measure.

\section{Supersymmetric $\mathbb{CP}^{n-1}$}\label{SUSYsec}
\setcounter{equation}0

Our next goal is to find the spectrum of  the supersymmetric $\mathbb{CP}^{n-1}$ model in a monopole background.
We will do it using the homogeneous coordinates $w^\alpha, \bar w^\alpha$ and their superpartners. It is instructive, however, to outline first the description of this model in terms of the inhomogeneous coordinates $z^j = w^j/w^0, 
\bar z^j = \bar w^j/\bar w^0$, which is mostly found in the literature.

\subsection{Inhomogeneous description}

To do so, we consider the 1-dimensional superspace $(t, \theta, \bar \theta)$ and introduce $n$ chiral superfields~$Z^j$ satisfying the condition $\bar D Z^j = 0$ and their conjugates $\bar Z^j$ satisfying $D \bar Z^j = 0$, where\footnote{See e.g. Chap. 7 and Chap. 9 of the book \cite{book}.} 
\be
\lb{covder}
D \ =\ \frac{\partial}{\partial \theta} - i\bar \theta \frac \partial{\partial t}\,, \qquad \qquad 
\bar D \ =\ - \frac \partial {\partial \bar \theta} + i \theta \frac   \partial{\partial t}
 \ee 
are the supersymmetric covariant derivatives. The component expansions of $Z^j$ and $\bar Z^j$ are
 \be
Z^j \ =\ z^j(t_L) + \sqrt{2}\, \theta \psi^j(t_L), \qquad \qquad
\bar Z^j \ =\ \bar z^j(t_R)  - \sqrt{2}\, \bar \theta \bar \psi^j(t_R),
 \ee
with $t_{L,R} = t \mp i \theta \bar \theta$. Here $\psi^j$ and $\bar \psi^j$ are the fermion superpartners of $z^j$ and $\bar z^j$.

The action of the model reads\footnote{ More known is the model involving real superfields and describing the geometry of the de Rham complex \cite{Wit82,Wit-Morse}. The model \p{act-Z} describes the {\it Dolbeault} complex and, in contrast to the model considered in Refs. \cite{Wit82,Wit-Morse},   allows one to include  gauge fields in the dynamics.}
 \be
\lb{act-Z}
S \ =\ \int d\theta d \bar \theta dt \left[ \frac 14 h_{j\bar k}(Z, \bar Z)
  D Z^j  \bar D \bar Z^k + {\cal R}(Z, \bar Z) \right]\,, 
 \ee
 where 
\be
\lb{R-Z}
{\cal R} \ =\ -\frac q 2   {\cal K}  = -\frac q2 \ln (1 + Z^j \bar Z^j)\,.
 \ee
The last term describes the monopole gauge field. Indeed, the corresponding component Lagrangian reads\footnote{For
a generic complex manifold, the Lagrangian also involves a quartic fermionic term. But for the K\"ahler manifolds it is absent.} 
 \be
\lb{L-comp-Z}
&&L \ = \  h_{j\bar k}\,\left[\dot{z}^j\dot{\bar{z}}^k +{i\over 2}\left(\psi^j  \dot{\bar{\psi}}^k-\dot{\psi}^j \bar{\psi}^k\right)\right]-{i\over 2}\,\left[\dd_j h_{l\bar k}
\,\dot z^l - \bar \dd_{k}  h_{j\bar l}\, \dot{\bar{z}}^l \right] \psi^j \bar{\psi}^k \nn
&&- A_j \dot z^j - \bar A_k \dot {\bar z}^k + i(\partial_j \bar A_k - \bar \partial_k A_j) \,\psi^j \bar \psi^k 
\ee
with 
 \be
\lb{A-CPn}
A_j = i \partial_j {\cal R} \ =\ -\frac {iq \bar z^j}{2(1+ z^l \bar z^l)},\, \qquad 
\bar A_k = -i\bar\partial_k {\cal R} \ =\ \frac {iq  z^k }{2(1+ z^l \bar z^l)},
 \ee
which coincides with \p{A-CP}.
The function ${\cal K}$ is nothing but the K\"ahler potential, $h_{j\bar k} \ = \ \partial_j \bar \partial_k {\cal K}$. 

By construction, the Lagrangian \p{L-comp-Z} is invariant  w.r.t. the supersymmetry transformations,
\be
&&\delta z^j = \epsilon\,\psi^j \,,\quad\quad   \delta \bar \psi^j=i\epsilon \dot{\bar z}^j\,, \nn
&&\delta \bar{z^j}=-\bar{\epsilon}\,\bar{\psi^j}\,,\quad\quad  \delta  \psi^j= -i\bar \epsilon  z^j\,.
\ee
The classical supercharges can be calculated using Noether's theorem. They read \cite{DiracSQM}
 \be
\lb{Qclass-pi}
&&Q_{\cal R} = \psi^j \left(\pi_j  - \frac i2 \partial_j h_{l\bar k} \psi^l \bar \psi^k  + i \partial_j {\cal R} \right  ), \nn
&&\bar Q_{\cal R} = \bar\psi^k \left(\bar\pi_k  + \frac i2 \partial_{\bar k} h_{j\bar l} \psi^j \bar \psi^k - i \bar \partial_k {\cal R}\right)\,,
 \ee
where the canonical momenta $\pi_j$ and $\bar \pi_k$ are obtained by differentiating the Lagrangian over $\dot{z}^j$ and $\dot{\bar z}^k$, while keeping $\psi^m$ and $\bar \psi^m$ fixed. 
 
 To define the corresponding supersymmetric quantum problem, one should\footnote{See \cite{howto,DiracSQM,book} for detailed explanations.}
\begin{enumerate}
\item To go over from the fermion variables $\psi^j$ and $\bar \psi^j$ carrying the world indices to the tangent space variables $\psi^a = e^a_j \psi^j, \ \bar \psi^a = \bar e^a_j \bar \psi^j$, where $e^a_j$ are the vielbeins. The advantage of this description is that  the classical Poisson bracket of the fermion variables becomes trivial $\{\psi^a, 
\bar \psi^b\}_{PB}  = i \delta^{ab}$, while the Poisson brackets like $\{\psi, z\}_{PB}$ vanish. The classical supercharges will then acquire the form
\be
\lb{Qclass-p}
&&Q_{\cal R} = e^j_c \psi^c \left(p_j -  i \omega_{a\bar b, j}  \psi^a \bar \psi^b  + i \partial_j {\cal R} \right)\,,  \nn
&&\bar Q_{\cal R} = \bar e^k_c \bar\psi^c \left(\bar p_k   - i  \omega_{a\bar b, \bar k}\psi^a \bar \psi^b  -   i \bar \partial_k {\cal R}\right)\,,
 \ee
where $p_j$ and $\bar p_k$ are obtained from the variation of the Lagrangian with keeping $\psi^c$ and $\bar \psi^c$ fixed  (they do not coincide with $\pi_j,  \bar \pi_k$) and $ \omega_{a\bar b, j},\, \omega_{a\bar b, \bar k}$ are the spin connections.

\item To go over from the classical supercharges to the quantum ones using the {\it Weyl} (symmetric) ordering prescription.
These quantum supercharges act in the Hilbert space with the flat measure.
 \item To make a similarity transformation and define the covariant quantum supercharges acting in the Hilbert space with the measure \p{norm-z}.
Such quantum supercharges are nilpotent. The anticommutator $\{\hat Q_{\cal R}, \hat {\bar Q}_{\cal R} \}$ gives the quantum Hamiltonian.

\end{enumerate}

Inspecting the quantum supercharges $\hat Q_{\cal R}$ and $\hat {\bar Q}_{\cal R}$ thus obtained, one can notice that

\begin{enumerate}
 \item The action of $\hat Q_{\cal R}$ on the wave functions $\Psi(z^m, \bar z^m; \psi^a)$ maps to the action of the nilpotent operator
\be
\dd_{\cal R} \ =\ \dd - \dd \left({\cal R} - \frac 14 \ln \det h \right) \wedge
\ee
of the twisted Dolbeault complex, where $\dd$  is  the operator of the exterior holomorphic derivative.  $\hat {\bar Q}_{\cal R}$ maps to the Hermitian conjugate operator.\footnote{See e.g. Propositions 1.4.23 and 1.4.25 in the book \cite{Nicolaescu} or Theorem 9.2 in the book \cite{book}).}

Thus, the Hilbert space of wave functions $\Psi(z^j, \bar z^j; \psi^a)$ where the quantum supercharges and Hamiltonian act may be mathematically described as the  set of all $(k, 0)$-forms on $\CP^{n-1}$ or else as 
\be
\lb{PiT}
  \mathcal{H}:={\cal L}_2\left(\Pi T\CP^{n-1}\right)\,,
  \ee
  where $\Pi T\CP^{n-1}$ is the $(1, 0$) tangent bundle\footnote{At every point $p\in \mathbb{M}$ of a complex manifold $\mathbb{M}$ we may decompose the complexified tangent space as $(T_p \mathbb{M})_{\mathbb{C}}=T_p^{(1, 0)} \mathbb{M}\oplus T_p^{(0, 1)} \mathbb{M}$.} to $\CP^{n-1}$ with fermionic fibers.

\item For  K\"ahler manifolds, the sum $\hat Q_{\cal R} + \hat {\bar Q}_{\cal R}$ is isomorphic to the {\it Dirac operator} \footnote{ To the best of our knowledge, this fact was first noticed by Hitchin \cite{Hitchin}. The operator \p{Dirac} is Hermitian. (Note the difference with the conventions of Ref. \cite{book} where the Dirac operator involved an extra factor  i and was anti-Hermitian.)} \lb{Dirac}
\be 
\lb{Dirac}
{\cal D}\!\!\!\!/  \ = \ e^M_A \gamma^A \left(-i\partial_M - \frac i4 \omega_{BC, M} \gamma^B \gamma^C +  A_M\right)\,,
 \ee
where $e^M_A$ are the (real) vielbeins, $\omega_{BC, M}$ are the spin connections and  $A_M$ is the gauge field.  For $\CP^{n-1}$, the latter  is given by 
Eq. \p{A-CP}. The Euclidean gamma matrices $\gamma^A$ satisfy the Clifford algebra:
\be
\lb{Clifford}
\gamma^A \gamma^B + \gamma^B \gamma^A \ =\ 2\delta^{AB}\,.
\ee
 The Dirac operator acts on the spinors\footnote{More exactly,  on the {\it bispinors} belonging to the representation $S_R \oplus S_L$, where $S_R$ and $S_L$ are two irreducible spinor representations of $Spin(D)$, $D$ being the real dimension of the manifold. In our example $D=2(n-1)$. 
One can observe in particular that the total number of independent components in  $\Psi(w,\bar w, \chi)$,
$$ N \ =\ \sum_{j=0}^{n-1} \, \left( \begin{array}{c} n-1 \\ j  \end{array} \right) \ =\ 2^{n-1}\,, $$
coincides with the number of bispinor components in the space of real dimension $D = 2(n-1)$. } living in $\mathbb{CP}^{n-1}$.  
\end{enumerate}

Note that the gauge field that twists the Dolbeault complex does not coincide with the physical gauge field that enters the Dirac operator and is determined by the second term in the integrand in Eq.\p{act-Z}, but is {\it shifted}. For $\mathbb{CP}^{n-1}$, this shift amounts to the shift of charge,
 \be
\lb{s-and-q}
q \ \to \ q - \frac n2 \ \stackrel {\rm def}= \ s.
 \ee
This is illustrated by the explicit expression for the covariant supercharge $\hat Q_{\cal R}$ with ${\cal R}$ written in Eq.\p{R-Z}:
 \be
\lb{supercharges-CP}
 \hat Q \ =\ -i\psi^k \left[ \frac \partial {\partial z^k} + \frac 12 \left(q - \frac n2 \right) \frac {\bar z^k}{1 + z^l \bar z^l} + \psi^a \hat {\bar \psi}^{\bar b} \, \omega_{a\bar b, k}   \right]\,.
 \ee

\subsection{Homogeneous description}

The K\"ahler potential on $\mathbb{CP}^{n-1}$ is expressed via homogeneous coordinates as 
\be 
{\cal K} \ =\ \ln (\bar w^\alpha w^\alpha) \ \equiv \ \ln {\cal X}\,. 
 \ee 
We now introduce $n$ chiral and $n$ antichiral superfields
\be
W^\alpha \ =\ w^\alpha(t_L) + \sqrt{2}\, \theta \chi^\alpha(t_L), \qquad \qquad
\bar W^\alpha \ =\ \bar w^\alpha(t_R)  - \sqrt{2}\, \bar \theta \bar \chi^\alpha(t_R)\,.
 \ee
and write the action as 
\be
\lb{act-W}
S \ =\ \int d\theta d \bar \theta dt \left[ \frac 14 h_{\alpha\bar\beta}(W, \bar W)
  D W^\alpha   \bar D \bar W^\beta+ {\cal R}(W, \bar W) \right]\,, 
 \ee
where ${\cal R} = -  (q/2) {\cal K}$ and
 \be 
h_{\alpha\bar\beta} \ = \ \dd_\alpha \bar \dd_\beta {\cal K}  \ =\ \frac {\Pi_{\alpha\bar\beta}}{\cal X}
\ee
with $\Pi_{\alpha \bar\beta} =\delta_{\alpha\bar\beta}  - {\bar{w^\alpha}w^\beta}/{\cal X}$.
It is convenient to write the component Lagrangian in the  following form:
  \be
\lb{L-comp-W}
L \ = \  h_{\alpha\bar\beta}\left(\dot{w}^\alpha\dot{\bar{w}}^\beta  - i \dot{\chi}^\alpha \bar{\chi}^\beta\right)
- i\dd_\alpha h_{\gamma\bar\beta}
\dot w^\gamma  \chi^\alpha \bar{\chi}^\beta  - 2\bar A_\alpha \dot {\bar w}^\alpha  + i(\partial_\alpha \bar A_\beta - \bar \partial_\beta A_\alpha) \chi^\alpha \bar\chi^\beta \nn
=  \frac {\Pi_{\alpha\bar\beta}}{\cal X}(\dot{w}^\alpha \dot{\bar{w}}^{\beta} - i  \dot{\chi}^\alpha \bar{\chi}^\beta)
+  \frac {i\chi^\gamma \bar{\chi^\beta} \dot w^\alpha} {{\cal X}^2} (\bar{w^\gamma}\Pi_{\alpha\bar\beta}+\bar{w^\alpha}\Pi_{\gamma\bar\beta}) 
-\ \frac {iq}{{\cal X}} w^\alpha  \dot {\bar w}^\alpha 
- \frac {q \Pi_{\alpha\bar\beta}}{\cal X} \chi^\alpha \bar \chi^\beta \,.
\ee
This complex Lagrangian differs from the real Lagrangian \p{L-comp-Z} (expressed in homogeneous coordinates) by a total derivative and brings about the same dynamics.
It  is invariant under the $\mathbb{C}^\ast$  gauge transformations:
\be
\lb{gauge-lambda}
w^\alpha \ \to \ \lambda(t) w^\alpha, \qquad \qquad \chi^\alpha \ \to \ \lambda(t) \chi^\alpha 
\ee
 (the terms $\propto \dot \lambda$  in the  transformed Lagrangian cancel). 
Due to supersymmetry, it is also invariant (as  is not difficult to check)  under the Grassmann-odd gauge transformations 
\be
\lb{gauge-eta}
 \delta \chi^\alpha \ =\ \eta(t) w^\alpha, \qquad \qquad \delta \bar \chi^\alpha =  \bar \eta (t) \bar w^\alpha\,.
\ee  
Bearing this in mind, one can present the supersymmetry transformations $\delta w^\alpha = \epsilon\,\chi^\alpha$ etc. in a $\mathbb{C}^\ast$-covariant form, 
\be
\label{SUSYhomtrans2}
&&\delta w^\alpha=\epsilon\,\chi^\alpha\,,\quad\quad  \delta \bar \chi^\alpha=i{\epsilon} {\cal D}\bar w^\alpha\,,\nn
&&\delta \bar{w^\alpha}=-\bar{\epsilon}\,\bar{\chi^\alpha}\,,\quad\quad   \delta {\chi^\alpha}=-i\bar \epsilon {\cal D}{w}^\alpha\,,
\ee
where the law of  $\mathbb{C}^\ast$ transformations for 
\be
{\cal D}w^\alpha \ =\ \dot w^\alpha - \frac {\bar w^\beta \dot w^\beta}{\cal X} w^\alpha \qquad {\rm and} \qquad 
{\cal D}\bar w^\alpha \ =\ \dot {\bar w}^\alpha - \frac {w^\beta \dot {\bar w}^\beta}{\cal X} \bar w^\alpha
 \ee
is the same as for $w^\alpha$ and $\bar w^\alpha$.

The canonical momenta for the Lagrangian \p{L-comp-W} are 
\be
\lb{can-mom}
&&p_{\alpha} = \frac{\dd L}{\dd \dot{\chi}^\alpha} \ = \ -\frac{i\,\Pi_{\alpha\bar\beta } \bar{\chi}^\beta}{\cal X}\,,\qquad
\bar \pi_{\alpha} = \frac{\dd L}{\dd \dot{\bar{w^\alpha}}} =\ \frac{\dot w^\beta \Pi_{\beta\bar\alpha}}{\cal X} 
- \frac {iq}{\cal X} w^\alpha\,,
\nn
&&\pi_{\alpha}=\frac{\dd \mathcal{L}}{\dd \dot{w^\alpha}}  =
{\Pi_{\alpha\bar\beta}\,\dot{\bar{w}}^{\beta}\over {\cal X}}+i\,\frac{\bar{w^\gamma}\Pi_{\alpha\bar\beta}+\bar{w^\alpha}\Pi_{\gamma\bar\beta}}{{\cal X}^2}  \chi^\gamma \bar{\chi^\beta} \,.
\ee

The following constraints hold: 
\be
\lb{super-cons}
{\cal C}_1 = w^\alpha \pi_\alpha + \chi^\alpha \, p_\alpha  \ = 0, \quad {\cal C}_2 = \bar w^\alpha\,\bar \pi_{\alpha} \ = \ -iq \,,\quad    {\cal C}_3 = w^\alpha\,p_{\alpha} = 0
\ee
 [one can check the validity of \p{super-cons} directly, substituting the momenta \p{can-mom} into \p{super-cons} and using $w^\alpha \Pi_{\alpha\bar\beta} = \bar w^\beta \Pi_{\alpha\bar\beta} = 0 $, but the true origin of the constraints are   the gauge symmetries \p{gauge-lambda} and \p{gauge-eta}]. 
The canonical Hamiltonian reads
\be
H \ =\ {\cal X} \left[\pi_\alpha  - \frac i{{\cal X}^2} (\bar w^\gamma \Pi_{\alpha\bar\beta} + \bar w^\alpha \Pi_{\gamma\bar\beta}) \chi^\gamma \bar \chi^\beta \right]  \left[\bar \pi_\alpha +  \frac {iq}{{\cal X}} w^\alpha  \right]
+ \frac {q}{\cal X} \Pi_{\gamma\bar\beta} \chi^\gamma \bar \chi^\beta \,.
 \ee
Taking into account the constraints \p{super-cons}, it reduces to a simple expression
\be
\lb{Ham-class-super-mon}
 H \ =\ {\cal X}  \pi_\alpha \bar \pi_\alpha + \bar w^\gamma \chi^\gamma p_\alpha \bar \pi_\alpha\,,
 \ee
one and the same for all $q$ !

The  supersymmetries~(\ref{SUSYhomtrans2})  bring about  conserved supercharges. For the complex Lagrangian \p{L-comp-W}, Noether's procedure gives
\be
\lb{Q-class-mon}
&&Q \ =\ \chi^\alpha \left(\pi_\alpha - \frac i2 \partial_\alpha h_{\beta\bar\gamma} \chi^\beta \bar \chi^\gamma \right) \ =\ \chi^\alpha \pi_\alpha, \nn
&&\bar Q \ =\ \bar \chi^\alpha \left( \bar \pi_\alpha +  \frac i2 \bar\partial_\alpha h_{\beta\bar\gamma} \chi^\beta \bar \chi^\gamma  - 2i \bar \dd_\alpha {\cal R} \right) \ =\ \bar \chi^\alpha\left( \bar \pi_\alpha + \frac {iqw^\alpha}{\cal X} \right) \ =\ i {\cal X} p_\alpha \bar \pi_\alpha\,.
 \ee
(The terms $\sim \dd h \chi \bar \chi$ vanish by symmetry: the factor $\dd_\alpha h_{\beta\bar\gamma}$ is symmetric under $\alpha \leftrightarrow \beta$ due to K\"ahlerian nature of the manifold while the factor $\chi^\alpha \chi^\beta$ is skew-symmetric;  the last formula for $\bar Q$ is valid in virtue of the expression for $p_\alpha$ in \p{can-mom} and the constraint $\bar w^\alpha \bar \pi_\alpha = -iq$).

For readers familiar with the 2D sigma model setup~\cite{Bykov2, Bykov3}, it is instructive to compare the Hamiltonian~(\ref{Ham-class-super-mon}) with the one appearing in 2D Gross-Neveu models. To match the two formulations, one rewrites the Hamiltonian as follows:
\be\label{Hamasymm}
H\ = \
\mathrm{Tr}\left[(w\otimes \pi+\chi\otimes p) (\bar{\pi}\otimes \bar{w}) \right]\,.
\ee
One notices the evident asymmetry between the holomorphic and anti-holomorphic pieces. On the contrary, in the study of  $\mathcal{N}=(2, 2)$ sigma models in~\cite{Bykov3}, the expression for the Hamiltonian is symmetric. The reason is that~(\ref{Hamasymm}) should arise from the dimensional reduction of an ${\mathcal{N}=(0, 2)}$ SUSY sigma model in 2D (see Ref.~\cite{Witten02SUSY} for a general introduction  and Refs.~\cite{CuiShifman, ChenShifman1, ChenShifman2} for applications to the $\CP^{n-1}$-model). Models with this amount of supersymmetry have not been formulated in Gross-Neveu language so far, but we expect that this could be done along the lines of the mechanical system discussed here.

\subsubsection{Quantization}
The canonical momenta \p{can-mom} become differential operators:
\be
\hat \pi_\alpha=-i{\partial \over \partial w^\alpha}, \quad  \hat {\bar \pi}_\alpha=-i{\partial \over \partial \bar{w^\alpha}}, \quad  p_\alpha = -i {\partial \over \partial \chi^\alpha} \,.
\ee
To write the quantum versions of the constraints \p{super-cons}, the supercharges \p{Q-class-mon} and the Hamiltonian \p{Ham-class-super-mon}, we should resolve possible ordering ambiguities in such a way that supersymmetry of the classical problem is preserved. This requirement  eliminates {\it almost} all ambiguities.

The quantum supercharge $\hat Q$ and the quantum constraint $\hat {\cal C}_3$ are restored without ambiguities. The requirement $\hat {\cal C}_3 \Psi = w^\alpha \hat p_\alpha \Psi = 0$ also eliminates the ambiguities in $\hat {\bar Q}$. The quantum Hamiltonian is then restored as the anticommutator $\{\hat Q, \hat {\bar Q}\}$. We obtain

\begin{empheq}[box = \fcolorbox{black}{white}]{align}
\lb{supercharges}
& \hat Q  \ =\ -i\chi^\alpha \frac \partial{\partial w^\alpha}, \quad \quad \quad
\hat {\bar Q} \ =\ -i {\cal X}  \frac {\partial^2}{\partial 
\chi^\beta \partial \bar w^\beta }\,, 
\\
\lb{super-H}
& \hat H \ =\  -{\cal X} \triangle -
\bar w^\alpha \chi^\alpha  \frac{\partial^2}{\partial \chi^\beta \partial \bar w^\beta} \,.
 \end{empheq}
The classical constraint ${\cal C}_1$ is equal to the Poisson bracket $\{Q, {\cal C}_3\}_{PB}$. To keep supersymmetry, we have to restore $\hat {\cal C}_1$ as the anticommutator $i\{\hat Q, \hat {\cal C}_3 \}$. Then $\hat {\cal C}_1 \ =\ w^\alpha \hat \pi_\alpha  + \chi^\alpha \hat p_\alpha$, which amounts to Weyl ordering of ${\cal C}_1$.

The only ambiguity not yet resolved dwells in $\hat {\cal C}_2$. We can write it as $\bar w^\alpha \hat {\bar \pi}_\alpha = s$ with any integer $s$ and supersymmetry will still be respected. In fact, this ambiguity reflects the fact that there are many different quantum problems with different integer (for even $n$) or half-integer (for odd $n$) charges $q$. If $q$ is defined as the physical charge of the gauge field as it appears in the Dirac operator and if we want to make contact with the literature where  inhomogeneous coordinates are mostly used, we have to resolve the ambiguity in $\hat {\cal C}_2$ using the Weyl  prescription:
 $$ \bar w^\alpha \bar \pi_\alpha \ \to \ \frac 12 (\bar w^\alpha \hat {\bar \pi}_\alpha + \hat {\bar \pi}_\alpha \bar w^\alpha ) \ =\ 
\bar w^\alpha \hat {\bar \pi}_\alpha - \frac {in}2\,.
 $$
This gives $\bar w^\alpha \hat {\bar \pi}_\alpha \ =\ -i( q -  n/2 ) \ =\ -is$ with $s$ defined in  \p{s-and-q}.

To reiterate, admissible wave functions $\Phi(w^\alpha, \bar w^\alpha; \chi^\alpha)$ satisfy the constraints
\begin{empheq}[box = \fcolorbox{black}{white}]{align}
\lb{superconstr}
&i\hat {\cal C}_1 \Psi \ =\ \left(w^\alpha \frac \partial {\partial w^\alpha} +  \chi^\alpha \frac \partial {\partial \chi^\alpha}\right) \Psi \ =\ 0, \nn
&i\hat {\cal C}_2 \Psi \ =\ \left(\bar w^\alpha \frac \partial {\partial \bar w^\alpha}\right )\Psi \ =\ s \Psi, 
\quad  i\hat{\cal C}_3 \Psi \ =\
 \left(w^\alpha \frac \partial {\partial\chi^\alpha}\right) \Psi \ =\ 0
\end{empheq}
with  $s = q - n/2$. The constraints are compatible with supersymmery due to 
\be
&&[\hat Q, \hat {\cal C}_1] = [\hat Q, \hat {\cal C}_2] \ = \ 0, \quad \{\hat Q, \hat {\cal C}_3\} = -i\hat {\cal C}_1, \nn
&&[\hat {\bar Q}, \hat {\cal C}_1] = [\hat {\bar Q}, \hat {\cal C}_2] =  \{\hat {\bar Q}, \hat {\cal C}_3\} \ =\  0.
 \ee
The supercharges \p{supercharges}  and the Hamiltonian \p{super-H} act on the Hilbert space with the inner product 
\be
\lb{super-measure}
\langle \Psi_1 | \Psi_2 \rangle = \int \bar \Psi_1(w, \bar w, \bar \chi) \Psi_2(w, \bar w, \chi)\, {\cal X}^n \,
\exp\left\{- {\cal X}- \frac{\chi^\gamma \bar \chi^\gamma}{{\cal X}}\right\} \prod_{\alpha = 0}^{n-1} dw^\alpha d \bar  w^\alpha d\chi^\alpha d\bar \chi^\alpha \,.
\ee
Then the supercharges $\hat Q$ and $\hat {\bar Q}$ are mutually conjugated\footnote{When flipping the derivative  $\partial/\partial w^\alpha$ in a matrix element $\langle \Psi_1 | \hat Q \Psi_2 \rangle$ and replacing simultaneously the operator $\chi^\alpha$ acting on the right by $-{\cal X} \partial/\partial \chi^\alpha$ acting on the left, the contribution coming from differentiating the measure is proportional to  $\hat {\cal C}_3 \Psi_1$ and  vanishes.} and the Hamiltonian is Hermitian.

 The wave functions in the sectors with different fermion numbers have the form
\be
\Psi^{F=0} \ =\ \Psi^{(0)} (w, \bar w),\qquad \Psi^{F=1} \ =\ \Psi^{(1)}_\alpha (w, \bar w) \chi^\alpha, \qquad \Psi^{F=2} \ =\ 
\Psi^{(2)}_{\alpha\beta}(w, \bar w) \chi^\alpha \chi^\beta,  \quad {\rm etc.}
 \ee
 These functions should obey the constraints \p{superconstr}. 
To make contact with the standard formulation of the supersymmetric   $\mathbb{CP}^{n-1}$ QM model, we have to resolve the constraints. In particular, the constraint 
  \be
\lb{C3}
 w^\alpha \frac \partial {\partial \chi^\alpha} \Psi \ =\ 0
  \ee
reduces the number of independent fermionic variables by 1, and the Hilbert space splits into $n$ sectors with the fermion numbers $F = 0,\ldots, n-1$. Indeed, an immediate corollary of the constraint \p{C3} is the absence of the sector $F=n$. The wave function in this sector would have the form
 \be
\Psi^{F=n}(w, \bar w, \chi) \ =\ A(w, \bar w) \, \varepsilon_{\alpha_1 \ldots \alpha_n} \chi^{\alpha_1} \cdots \chi^{\alpha_n}\,.
\ee
Substituting this in \p{C3}, we obtain
 \be
 A(w, \bar w)\,  \varepsilon_{\alpha_1 \ldots \alpha_n} \chi^{\alpha_1} \cdots \chi^{\alpha_{n-1}}w^{\alpha_n}  = 0 \quad \Rightarrow \quad A(w, \bar w) = 0\,.
\ee

In other words, the Hilbert space is still \p{PiT} despite the presence of an extra fermionic variable\footnote{From the mathematical standpoint, constraints~(\ref{superconstr}) correspond to a description of the holomorphic tangent bundle to $\CP^{n-1}$  in terms of the tautological bundle, cf.~\cite{Lam} or~\cite{Bykov3}. }.

\subsection{The top/bottom fermion sectors and the zero modes}

Before considering the sector of arbitrary fermion number $F$, let us analyze the sector of minimal and maximal fermion numbers, i.e. $F=0$ and $F= n-1$. First of all, the supersymmetric Hamiltonian simplifies considerably in  these sectors. Besides, it is precisely these sectors that potentially contain zero modes, which are of special geometric significance.

\subsubsection{Minimal fermion number: $F=0$}

In this sector, only the first term in the Hamiltonian \p{super-H} is relevant,
 and the problem is almost the same as in the bosonic case modulo a slight modification of the Hamiltonian\footnote{Now it has the form $-{\cal X}\triangle$ even in the presence of the monopoles, whereas our previous bosonic choice was the Bochner laplacian \p{Bochner}.} and the constraints. The wave functions can now be presented in the form 
\be
\lb{Psi-F0-s+}
\Psi^{F=0}(w^\alpha, \bar w^\alpha) \ =\ (\bar w^0)^s \Phi(z^j, \bar z^j) \quad {\rm with} \quad  z^j = \ \frac {w^j}{w^0}.
  \ee 
We derive: 
\be
\hat H_w \Psi \ =\ (\bar w^0)^s \hat H_z \Phi \ =\  (\bar w^0)^s (1 + \bar z^l z^l ) \left[ s z^j  \frac {\partial \Phi}{\partial  z^j } - \frac {\partial^2 \Phi}{\partial z^j \partial \bar z^k}(\delta^{jk} + z^j \bar z^k) \right].
\ee
As follows from \p{super-measure}, the wave functions $\Psi^{F=0}(w^\alpha, \bar w^\alpha)$  are normalized with the measure \p{norm-w}. This implies that the functions $\Phi(z^j, \bar z^j) $ are normalized with the measure
$$
\sim \frac {\prod_j dz^j d\bar z^j}{(1 + z^l \bar z^l)^{n+s}}\,. $$
To go over to the functions normalized with the standard measure \p{norm-z} and to the Hamiltonian acting on them, we have to perform a similarity transformation,
 \be
\lb{simil}
\Phi'  =   (1 + z^l \bar  z^l)^{-s/2} \Phi, \qquad \hat H_z'  = (1+ z^l \bar  z^l)^{-s/2} \hat H_z (1 + z^l \bar  z^l)^{s/2}.
 \ee
We  derive
\be
\lb{Ham-z-F0}
\hat H_z' \ =\ (1 +  z^l \bar z^l )\left[ -  (\delta^{jk} + z^j \bar z^k) \partial_j \bar \partial_k  + \frac s2 (z^j \partial_j -   
 \bar z^j  \bar \partial_j ) \right] + \frac {s^2}4 z^l \bar z^l + \frac s2(1-n).
\ee
The Hamiltonian \p{Ham-z-F0} has a form similar to the bosonic Hamiltonian \p{H-Phi} where one should replace $q$ by $s$ and add the  constant 
$s(1-n)/2$.

The spectrum is\footnote{For $n=2$ it was derived in \cite{Kim-Lee}.}
\be
E^{F=0}(n,L,L') \ =\ L(L' + n-1),
\ee
with $L \geq 0$ and  $ L' = L+s \geq 0$.

The multiplicity is given by Eq.\p{dim-q}. For $s \geq 0$, the spectrum includes 
\be
\lb{IW}
\#_0^{F=0} \ = \  \left(\begin{array}{c} s+n-1\\ n-1 \end{array}  \right) 
\ee
states with zero energy. They are the vacuum states of the full supersymmetric Hamiltonian. Their number \p{IW}, related to the {\it Witten index} of the $\mathbb{CP}^{n-1}$ SQM system, is well known~\cite{Kirchberg} (see also \cite{Ivanov}) and  can as well be easily computed  in our approach. For $s \geq 0$, the wave functions \p{Psi-F0-s+} describing the states of minimal energy are linear combinations of the monomials 
\be
\lb{monoms}
\Psi_{\rm vac} \ =\ \prod_{\alpha =0}^{n-1} (\bar w^\alpha)^{k_\alpha} \qquad {\rm with} \quad  \sum_\alpha^{n-1} k_\alpha = s\,.
 \ee
In other words, these are homogeneous polynomials of degree $s$ in $n$ variables. The dimension of this vector space is given by the binomial coefficient in Eq.\p{IW}.

The geometric meaning of this index is not the standard Euler characteristic of the manifold, as is the case for the SQM sigma model describing the de Rham complex \cite{Wit-Morse}, but rather  the {\it holomorphic} Euler characteristic of the vector bundle ${\cal E}$ describing the wave functions\footnote{In our applications, $\mathcal{E}=\mathcal{O}(s)$ is a power of the tautological line bundle.}, 
      \be
 \chi(\CPn, {\cal E}) \ =\ \sum_{i=0}^{n-1} (-1)^i\,{\rm dim}_{\mathbb{C}}H^i(\CPn, {\cal E})\,, 
 \ee  
which is  related to the  Hirzebruch-Riemann-Roch theorem \cite{Hirzebruch} (see also \cite{Alvarez,Smi-Hirz}).

If $s <0$, the states with zero energy in the sector $F=0$ are absent.

\subsubsection{Maximal fermion number: $F = n-1$.}

Consider now the sector $F = n-1$. The constraint \p{C3} dictates
\be
\lb{F:n-1}
\Psi^{F=n-1}(w, \bar w, \chi) \ =\ B(w, \bar w)\, \varepsilon_{\alpha_1 \ldots \alpha_n} \chi^{\alpha_1} \cdots \chi^{\alpha_{n-1}} w^{\alpha_n} \,.
\ee
Two other constraints give
\be 
\lb{B1B2}
w^\alpha \frac \partial {\partial w^\alpha} B \ =\ -nB , \qquad   \bar w^\alpha \frac \partial {\partial \bar w^\alpha}B\ =\ sB\,.
\ee
Consider the action of the Hamiltonian \p{super-H} on the function \p{F:n-1}. There are two contributions: {\it (i)} the term $\sim - {\cal X} \triangle  B$; {\it (ii)} the cross term. The latter gives
\be
- \bar w^\beta \left[ \left(w^\beta \frac \partial {\partial w^\gamma} + \chi^\beta \frac \partial {\partial \chi^\gamma }   \right)   \varepsilon_{\alpha_1 \ldots \alpha_n} \chi^{\alpha_1} \cdots \chi^{\alpha_{n-1}} w^{\alpha_n}  \right] \frac {\partial B}{\partial \bar w^\gamma}
\ee
Simple algebraic manipulations allow one to observe: 
$$[\cdots]^\beta_\gamma \ =\ \delta^\beta_\gamma \, \varepsilon_{\alpha_1 \ldots \alpha_n} \chi^{\alpha_1} \cdots \chi^{\alpha_{n-1}} w^{\alpha_n}\,.$$
It follows that 
\be
{\rm \it The\ cross\ term} \ =\ - \bar w^\beta \frac {\partial B}{\partial \bar w^\beta} \,  \varepsilon_{\alpha_1 \ldots \alpha_n} \chi^{\alpha_1} \cdots \chi^{\alpha_{n-1}} w^{\alpha_n} \ =\ - sB \, \varepsilon_{\alpha_1 \ldots \alpha_n} \chi^{\alpha_1} \cdots \chi^{\alpha_{n-1}} w^{\alpha_n}
 \ee
in virtue of the second constraint in \p{B1B2}. 

Thus, we are in a position to solve the spectral problem
\be
\lb{Ham-w-n-1}
(- {\cal X} \triangle - s) B \ =\ EB
\ee
with the constraints \p{B1B2}. We write
\be
B(w, \bar w) \ =\ \frac {(\bar w^0)^s} {(w^0)^n} \, \Phi(z^j, \bar z^j)\,.
\ee

 Acting on this by the Hamiltonian in \p{Ham-w-n-1} and performing a proper similarity transformation to arrive at the Hilbert space with the measure \p{norm-z}, 
  \be
\lb{simil1}
\Phi'  =   (1 + z^l \bar  z^l)^{(n-s)/2} \Phi, \qquad \hat H_z'  = (1+ z^l \bar  z^l)^{(n-s)/2} \hat H_z (1 + z^l \bar  z^l)^{(s-n)/2},
 \ee
we derive the Hamiltonian 
\be
\lb{Ham-z-n-1}
\hat H_z =  (1 +  z^l \bar z^l )\left[ \frac {s+n}2 ( z^j \partial_j  -
 \bar z^j \bar \partial_j ) -  (\delta^{jk} + z^j \bar z^k) \partial_j \bar \partial_k \right]\nn
 + \frac {(s+n)^2}4 \bar z_k z_k +  \frac {s+n}2(n-1).
\ee
It is instructive to rewrite this Hamiltonian and also the Hamiltonian in \p{Ham-z-F0} in terms of the monopole charge $q$. We derive
 \be
\hat H'_z(F=0) \ =\ (1 +  z^l \bar z^l )\left[ \frac {q-n/2}2  (z^j \partial_j  - \bar z^j \bar\partial_j) -  (\delta^{jk} + z^j \bar z^k) \partial_j \bar \partial_k \right] \nn
+ \frac 14 \left(q - \frac n2 \right)^2
 z^l\bar z^l -   \frac {q-n/2}2 (n-1)\,, \nn 
\hat H'_z(F=n-1) \ =\ (1 +  z^l \bar z^l )\left[ \frac {q+n/2}2  ( z^j \partial_j   - \bar z^j \bar\partial_j) -  (\delta^{jk} + z^j \bar z^k) \partial_j \bar \partial_k \right] \nn
+ \frac 14 \left(q + \frac n2 \right)^2
 z^l\bar z^l  +  \frac {q+n/2}2 (n-1)\,.
\ee
We observe that these Hamiltonians differ from one another by the sign of $q$ and the interchange $z \leftrightarrow \bar z$. In other words, the spectrum of the Hamiltonian in the sector $F=n-1$ is the same as in the sector $F=0$ for the system with the opposite monopole charge.  Also the wave functions $\Phi'$ there are the same up to the exchange $z \leftrightarrow \bar z$. Accordingly, 
\be
B(w, \bar w, q) \ =\ \Psi^{F=0}(\bar w, w, -q) {\cal X}^{q- n/2}\,.
\ee
In particular, for negative $q$,  the system has no zero modes in the sector $F=0$, but they may appear in the sector $F = n-1$. This happens if $q \leq -n/2$  (or $s \leq n$). 
 Their wave functions read
 \be
\Psi_{\rm vac} \ =\ \frac {\varepsilon_{\alpha_1\ldots \alpha_n}\chi^{\alpha_1} \cdots \chi^{\alpha_{n-1}} w ^{\alpha_n}}{{\cal X}^{n/2-q}} \prod_{\alpha=0}^{n-1} (w^\alpha)^{k_\alpha}  \qquad {\rm with} \qquad 
\sum_\alpha^{n-1} k_\alpha = -q - \frac n2 \,.
  \ee
There are 
\be
\lb{number-modes}
\#^{F = n-1}_0 \ =\ \left(   \begin{array}{c} -s-1 \\ n-1 \end{array} \right)
 \ee
such modes. 
In the range $-n/2 < q < n/2$, the zero modes are absent, supersymmetry is spontaneously broken, and the Witten index,
 \be
\lb{IW-def}
  I_W \ =\ n_B^{(0)} - n_F^{(0)}  
 \ee
($n_B^{(0)}$ is the number of zero modes with even fermion charge and $n_F^{(0)}$ is the number of zero modes with odd fermion charge), is equal to zero.
For $q \geq n/2$ or $q \leq -n/2$, the zero modes are present and 
  their number  is given by the 
binomial coefficient in Eq.\p{IW} with $s = q - n/2$ for positive $q$ and $s \to -n-s = -(q + n/2)$ for negative $q$. If $n$ is odd, the Witten index is equal to this number both for positive and negative $q$, whereas, for even $n$ and $q \leq -n/2$, it acquires the extra minus factor.\footnote{Note in passing that basically the same formula describes the Witten index in the 3-dimensional supersymmetric Yang-Mills-Chern-Simons theory, and such a coincidence is not accidental \cite{3d}.}

\subsection{The full spectrum}

The zero modes of the supersymmetric Hamiltonian have been well known, but our technique allows one to obtain the whole spectrum.
We have already done so in the sectors $F=0$ and $F =n-1$. For $\mathbb{CP}^1$, there is nothing else.

The results of the previous section also allow one to find the full spectrum  for $\mathbb{CP}^2$. The states in the intermediate sector $F=1$  are superpartners of  the states with either $F=0$ or $F=2$, and their wave functions can be found by acting with $\hat Q$ or $\hat {\bar Q}$ on the corresponding states. There are therefore two different chains of $F=1$ states: the states $\Psi_+^{F=1}$ representing the upper components of the supersymmetric doublets and  the states $\Psi_-^{F=1}$ representing the lower components. The energies of the lowest states in the sectors $F=0,1,2$ are given in Table~\ref{spec-CP2}.

\begin{table}[h!]
\caption{Energies of the    lowest states on $\mathbb{CP}^2$ }
\label{spec-CP2}
\begin{center}
\tabcolsep=0.5em
\begin{tabular}{|l||c|c|c|c|c|c|c|}
\hline
 q &-5/2 & -3/2 & -1/2 & 1/2 & 3/2 & 5/2\\
\hline 
 $\Psi^{F=0}$ &  8  & 6 & 4 & 2 & 0 & 0\\
 \hline
$\Psi_+^{F=1}$ &  8  & 6 & 4 & 2 & 3 & 4 \\
\hline
 $\Psi_-^{F=1}$ &  4  & 3 & 2 & 4 & 6 & 8 \\
\hline
 $\Psi^{F=2}$ &  0  & 0 & 2 & 4 & 6 & 8 \\
 \hline
\end{tabular}
\end{center}
\end{table}

For ${n \geq 4}$, this procedure gives {\it some} states in the sectors $F=1$ and $F = n-2$,   but not all of them. For example, for $\mathbb{CP}^3$, there are supersymmetric doublets involving  the states in the sectors $F=0,1$ and $F=2,3$, but also the doublets with the states in the sectors $F = 1,2$, which cannot be ``reached" from the leftmost and the rightmost sectors.  Still, the wave functions of the states in the ``intermediate" sectors $F = 1, \ldots, n-2$ can be found if we concentrate on the {\it lower} components $\Psi_-$ of the supersymmetric doublets.

Our general ansatz for the wave functions is a straightforward generalization of the bosonic one~(\ref{Psi-Ans-q}):
\begin{empheq}[box = \fcolorbox{black}{white}]{align}
\lb{PsiFermMon}
\quad \Psi(w, \bar w, \chi) \ =\ \frac 1{{\cal X}^{M}} A_{\alpha_1 \ldots \alpha_L| \bar\beta_1 \ldots \bar\beta_{L'}| \gamma_1 \ldots \gamma_F} w^{\alpha_1} \cdots w^{\alpha_L} \, 
 \bar w^{\beta_1} \cdots \bar w^{\beta_{L'}}\,\chi^{\gamma_1}\cdots \chi^{\gamma_F}\quad
\end{empheq}
The bosonic constraints imply
\be
M=L+F,\quad\quad L'-M=s\,.
\ee
Consider the states $\Psi_-$, which by definition are annihilated by the action of the supercharge $\hat{\bar Q}$ in Eq.\p{supercharges}: 
\be
\lb{QbarPsi}
\hat {\bar Q} \Psi_- = 0\,.
 \ee
Hence they are also annihilated by the action of the second term in the Hamiltonian \p{super-H} and only the first term is relevant. We can then repeat the reasoning of Theorem 1 and conclude that, for the function \p{PsiFermMon} satisfying the condition \p{QbarPsi} to be an eigenstate of the Hamiltonian,   the tensor $A$ 
should satisfy the $\alpha\beta$ tracelessness condition,\footnote{This  is not so for the states  $\Psi_+$, which are not annihilated  by $\hat {\bar Q}$ and the second term in the Hamiltonian \p{super-H} cannot be disregarded in the action $\hat H \Psi_+$. This means in particular that the tensor $A$ for these states need not satisfy the tracelessness condition \p{albet-trace}.  And it never does! (See Appendix~\ref{updoubletapp} for a proof.)}
 \be
\lb{albet-trace}
A_{\delta \ldots \alpha_L| \bar\delta \ldots \bar\beta_{L'}| \gamma_1 \ldots \gamma_F} \ =\ 0\,.
 \ee
The energies can then be found by a direct action of the Hamiltonian $\hat H \equiv -{\cal X} \triangle$ on  \p{PsiFermMon}, where the fermion factors play the role of ``spectators". We derive
\begin{empheq}[box = \fcolorbox{black}{white}]{align}
\lb{E-LL'F}
E_-(n,L,L',F) \ =\ (L+F)(L' -F + n -1)\,,
\end{empheq}
where $F = 0,\ldots, n-2$, $L' = L+F+s \geq 0$. When $F=0$, the integer $L$ may also be positive or zero, but for nonzero $F$, $L$ is necessarily positive --- this is a corollary of the fermionic constraint $w^\alpha \partial \Psi/ \partial \chi^\alpha = 0$. It cannot be fulfilled if $\Psi$ has no $w$ factors. 
The formula \p{E-LL'F} describes also the zero modes in the sector $F=n-1$, which exist when $L' = 0$ and $s+n \leq 0$. These are  the only states in the sector $F=n-1$ that are annihilated by $\hat {\bar Q}$.

The  eigenstates of the Hamiltonian belong to  degenerate in energy $SU(n)$ multiplets. What remains is to describe these multiplets explicitly.
  \begin{figure}
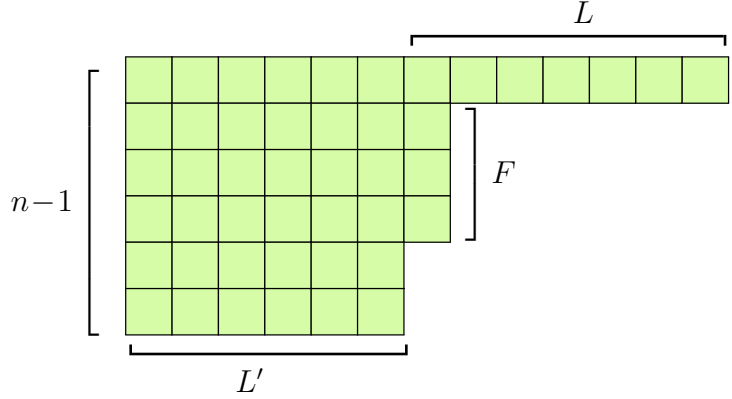

\centering
\begin{Overpic}[]{
\ytableausetup{boxsize=0.6cm}
\ytableaushort
{}
* [*(lbluecol)]{13,7,7,7,6,6}
}
\put(-15,20){$n\!-\!1 \left[
\begin{array}{rrr}
\\
\\
\\
\\
\\
\\
\\
\end{array}
\right. $}
\put(49,46){$\overbracket[1pt][0.6ex]{\hspace{4.2cm}}_{}$}
\put(75,50){$L$}
\put(4,-1){$\underbracket[1pt][0.6ex]{\hspace{3.7cm}}_{}$}
\put(21,-9){$L'$}
\put(53,24.5){$ \left.
\begin{array}{rrr}
\vspace{-0.2cm}\\
\\
\\
\\
\end{array}
\right] F $}
\end{Overpic}
\vspace{0.5cm}
\caption{Young diagram for the multiplet of supersymmetric monopole harmonics.}
\label{young2}
\end{figure}

\begin{thm}\lb{Young-F}
\vspace{0.3cm} 
The states $\Psi_-$ defined by~(\ref{QbarPsi}), which additionally satisfy the constraints~(\ref{superconstr}) and are eigenstates of the  Hamiltonian \p{super-H}, fall into one of the following classes:
\begin{itemize}
\item $F \neq 0$ and $F\neq n-1$,  $L\geq 1, L'=L+F+s\geq 0$
\item $F = 0,  L\geq 0, L'=L+s\geq 0$
\item $F=n-1$, $L'=0$, $L=-(n-1+s)\geq 1$
\end{itemize}
The energies of the states are given by~\p{E-LL'F}. Their $SU(n)$-representations are described by the Young diagram shown in Fig.~\ref{young2}.

For $F=0, L=0$ or $F=n-1, L'=0$ these states  are zero modes.

\vspace{0.3cm}
Unless the state $\Psi_-$ is a zero mode, it has a superpartner $\Psi_+$ with the same energy and furnishing the same representation of $SU(n)$. The states $\Psi_{\pm}$ exhaust all the eigenstates of the Hamiltonian.

\vspace{0.2cm}

\end{thm}

\begin{proof}

An arbitrary function \p{PsiFermMon} comprises a symmetric tensor in $\bar{w}$'s, another symmetric tensor in $w$'s and a skew-symmetric tensor in the $\chi$'s.   In terms of $SU(n)$ representations, this corresponds to the tensor product:
\be\label{tripletensprod}
\ytableausetup{boxsize=0.3cm}
\ytableaushort
{}
* [*(lbluecol)]{6,6,6,6,6,6}
\otimes  \ytableaushort
{} * [*(lbluecol)]{7}
\otimes \ytableaushort
{} * [*(lbluecol)]{1,1,1} \,.
\ee
This product is  expanded into a sum of many irreps. To pick out the representation that describes the eigenstates of our Hamiltonian, we should take into account the constraints \p{superconstr}, the condition \p{QbarPsi} and the traceless condition \p{albet-trace}.

Let us start with the last two factors in \p{tripletensprod}. The condition $\hat{\cal C}_3 \Psi = 0$ dictates that, if one replaces one of $\chi$'s with a $w$, the result should be zero, i.e. the fermions are anti-symmetrized with the bosons. As a result, this condition ensures that only the following representation survives in the tensor product:
\be
\ytableaushort
{} * [*(lbluecol)]{7}
\otimes \ytableaushort
{} * [*(lbluecol)]{1,1,1} \quad \mapsto \quad 
\ytableaushort
{} * [*(lbluecol)]{7,1,1,1}
\ee
It remains to compute the tensor product of this hook with the big rectangular box in~(\ref{tripletensprod}).

In computing the tensor product, we will take into account the condition $\hat {\bar Q} \Psi = 0$. Let us see what additional constraints on the wave function it  imposes. The derivative $\dd/\dd \bar w^\alpha$ in $\hat{\bar Q}$ may act on the numerator or on the denominator of \p{PsiFermMon}. The action on the numerator gives $\beta\gamma$ traces. The action on the denominator gives the structure 
$w^\alpha \dd \Psi/\dd \chi^\alpha$; this contribution vanishes due to the constraint $\hat {\cal C}_3 \Psi  = 0$. 
As a result, the tensor $A$ must satisfy, on top of the condition \p{albet-trace}, also the condition
 \be
\lb{betgam-trace}
A_{\alpha_1 \ldots \alpha_L| \bar\delta \ldots \bar\beta_{L'}| \delta \ldots \gamma_F} \ =\ 0\,.
 \ee

 Now consider the product  {\sl box} $\times$ {\sl hook}. By the rules for the products of Young diagrams~\cite{Georgi} we should take the small boxes from the  hook and glue them to the box either on the right or on the bottom. While doing so, we should abide by certain rules,  which however are irrelevant for us at the moment. As a result, we obtain

\be\label{tensprodyoung}
\ytableausetup{boxsize=0.2cm}
&& \ytableaushort
{}
* [*(lbluecol)]{6,6,6,6,6,6} 
\otimes 
\ytableaushort
{} * [*(bluecol)]{7} * [*(pinkcol)]{1,1,1,1}\quad = \quad \ytableaushort
{}
* [*(lbluecol)]{6,6,6,6,6,6}  * [*(bluecol)]{13} * [*(pinkcol)]{7,7,7,7} \quad \oplus  \quad
\ytableaushort
{}
* [*(lbluecol)]{6,6,6,6,6,6}  * [*(bluecol)]{12} * [*(pinkcol)]{7,7,7,7} * [*(bluecol)]{1,1,1,1,1,1,1} \quad \oplus \quad \\ \nonumber
&& \oplus \quad \ytableaushort
{}
* [*(lbluecol)]{6,6,6,6,6,6}  * [*(bluecol)]{11} * [*(pinkcol)]{7,7,7,7} * [*(bluecol)]{1,1,1,1,1,1,2}
\quad \oplus \quad \cdots \quad \oplus \quad 
\ytableaushort
{}
* [*(lbluecol)]{6,6,6,6,6,6}  * [*(bluecol)]{13} * [*(pinkcol)]{7,7,7} * [*(pinkcol)]{1,1,1,1,1,1,1} \quad \oplus \quad
\ytableaushort
{}
* [*(lbluecol)]{6,6,6,6,6,6}  * [*(bluecol)]{12} * [*(pinkcol)]{7,7,7} * [*(bluecol)]{1,1,1,1,1,1,1} * [*(pinkcol)]{2,2,2,2,2,2,2} \quad \oplus \quad \cdots
\ee
Note however that placing a blue or red box in the last ($n$th) row means taking a trace, either an $\alpha\beta$ trace in case of a blue box, or a 
$\beta\gamma$ trace in case of a red one. These traces vanish, however, and all the  representations where at least one of the boxes is glued on the bottom drop out. Then the only representation that remains (assuming $F\leq n-2$) is the one where the whole hook is glued on the right of the large rectangle,  the first one in the r.h.s. of~(\ref{tensprodyoung}), which coincides with  the one shown in Fig.~\ref{young2}.

In the simplest case $F=0$, the Young diagram acquires the form (\ref{DynkinBosonic}), the same as in the nonsupersymmetric case, so that only the left and right nodes of the Dynkin diagram carry non-zero labels. If, in addition, $L=0$, the corresponding wave function~(\ref{PsiFermMon}) has no dependence on $w^\alpha$ and is therefore annihilated by $\hat Q$. Therefore the state with $F=0$ and $L=0$ is a zero mode. 

Finally, in the case $F=n-1$ there are `too many' red boxes in~(\ref{tensprodyoung}), and one of them would be inevitably glued to the bottom of the rectangle, so that the resulting wave function would be zero by the constraints. The only exception is when there is no rectangle altogether, i.e. $L'=0$, which is the third case of the theorem. Since the supercharge $\hat Q$ raises the fermion number by one, and there are no admissible states with $F=n$, one concludes that these states are annihilated by both $\hat Q$ and $\hat {\bar Q}$, implying that they are zero modes.
The corresponding Young diagram is a row including $L-1 = -(s+n)$ boxes. The dimension of this representation is given by Eq. \p{number-modes}.

So far we discussed the states $\Psi_-$ representing lower components of the supersymmetric doublets. Whenever $\Psi_-$ is \emph{not} a zero mode, one can construct its superpartner $\Psi_+:= \hat Q \Psi_-$. By supersymmetry, it has the same energy $E\neq 0$ and   furnishes
the same $SU(n)$ multiplet.\footnote{It may be worth reminding that the states $\Psi_+$ and $\Psi_-$ annihilated either by $\hat Q$ or by $\hat {\bar Q}$ constitute the complete spectrum of a supersymmetric quantum system.} In Appendix~\ref{updoubletapp}, we will illustrate this by analysing the relevant Young diagrams. If $\Psi_-$ is a zero mode, it is annihilated by both supercharges, so that it has no superpartner.
\end{proof}

\vspace{0.3cm}
It follows from the proof that the only representation that survives in the tensor product~\p{tensprodyoung}, once all constraints are taken into account, is the one whose highest weight is equal to the sum of the highest weights of the factors.

Generically (for $F>0$) it has the form
\be
\quad a_{i = 1,\ldots, n-1} = (L-1, \underbracket[1pt][0.6ex]{\;0, \;\ldots, \;0,}_{F-1} \;1, \underbracket[1pt][0.6ex]{\;0, \;\ldots, \;0,}_{n-3-F} \;L') \quad
\ee
One can then compute the dimensions of the representations either by the hook formula from the Young diagram, or by the Weyl formula \p{Weyl} using the Dynkin labels.

\subsection{Dirac operator and its spectrum}\label{Diracsec}

 As we have already mentioned on p. \pageref{Dirac},   
 on a K\"ahler manifold, the Dirac operator is isomorphic to the sum of the supercharges (which are in turn isomorphic to the exterior holomorphic derivative in the Dolbeault complex and its conjugate):
\be
{\cal D}\!\!\!\!/  \ \simeq \ \hat Q + \hat {\bar Q}
\ee
Then the supersymmetric Hamiltonian  $H=\{Q, \hat {\bar Q}\}_+$ is isomorphic to the square of the Dirac operator\footnote{In three dimensions, this operator, furnished with the factor $1/2m$, is known to physicists as the 
{\it Pauli Hamiltonian} describing the motion of electron in external magnetic field.}
 $(/\!\!\!\!D)^2$.
In the main part of Sect. 3, we have determined the eigenvalues $E$ of the supersymmetric Hamiltonian on $\CP^{n-1}$ in a field of a $\mathbb{CP}$ monopole, and this immediately gives us the Dirac eigenvalues: 
\be
\lambda_D \ =\  \pm  \sqrt{E}
 \ee
The eigenfunctions of $/\!\!\!\!D$  are the bispinors isomorphic to the linear combinations  of the functions $\Psi_-$ and  $\Psi_+$  found above. The zero modes of \,$/\!\!\!\!D$, if they exist, map to either $\Psi_-$ or $\Psi_+$, depending on the sign of the monopole charge, and nonzero modes map to  
  \be
\Psi_1  = \Psi_- + \Psi_+\,,\quad {\rm and} \quad \Psi_2 = \Psi_--\Psi_+
\ee
with $\Psi_+ =   /\!\!\!\!D \Psi_-/\sqrt{E}$.

 The eigenstates of $/\!\!\!\!D$ form the $SU(n)$ multiplets described by the Young diagram in Fig.~\ref{young2}.

While preparing our paper for publication, we came across the paper \cite{Dolan}, where 
the spectrum of the Dirac operator on  {\it fuzzy} $\CP^{n-1}$ was found (see also the much earlier paper \cite{Seifarth} where this spectrum was determined on $\CP^{2k+1}$ in the absence of the gauge field).  Their expression for the eigenvalues of $(/\!\!\!\!D)^2$ coincides with our formula \p{E-LL'F} and the $SU(n)$ multiplets are described by the same Young diagram.

\section{Conclusions and outlook}

In the first half of the  paper we reproduced the known results for the spectrum of the Laplacian on $\CP^{n-1}$ using the {\it homogeneous} coordinates $w^\alpha$ to describe its geometry. We did so also for the Bochner Laplacian \p{Bochner} which includes the background gauge field \p{CP-monopole-w} representing   a straightforward generalization of the Dirac monopole on $\CP^1\simeq S^2$. 

The $\mathbb{C}^{\ast}$ invariance of the metric \p{metr-homogen} is reflected in the gauge symmetry of the corresponding dynamical system. In the quantum setup, the symmetry is implemented via the constraints \p{cons-bos-q}. We use a natural polynomial ansatz \p{Psi-Ans-q} for the wave functions, in which case the constraints can be easily resolved. The eigenstates form degenerate multiplets corresponding to the irreducible representations of $SU(n)$ with the  Young diagrams drawn in Fig. 1.

In Sect.~3, we   applied the same machinery to the $\mathcal{N}=2$ supersymmetric quantum-mechanical 
$\CP^{n-1}$ sigma model involving chiral superfields and describing the Dolbeault complex. In contrast to a more widely known  $\mathcal{N}=4$ model with real superfields which describes the K\"ahler -- de Rham complex, the Dolbeault complex can be twisted by a background gauge field. Choosing the background in the form of the $\CP$ monopole \p{CP-monopole-w}, we resolved the spectral problem, finding the eigenvalues and the wave functions expressed in the form 
\p{PsiFermMon}. 

The twisted Dolbeault complex is isomorphic to the Dirac complex, and our results for the eigenvalues and the structure of $SU(n)$ multiplets conform with the results of Ref. \cite{Dolan}.

The use of homogeneous coordinates brings about considerable simplifications, and one can apply it for studying other problems. One possibility is to consider in these terms the K\"ahler -- de Rham $\mathcal{N}=4$ supersymmetric quantum mechanics.   It arises after performing the dimensional reduction of $\mathcal{N}=(2, 2)$ two-dimensional sigma models studied in ~\cite{Bykov3}. Besides, this method could be  applied to solve similar spectral problems on more complicated complex manifolds such as homogeneous manifolds or even to the wider class of  models with quiver phase spaces  (see Ref.~\cite{Bykov2} for an overview).

\vspace{1cm}
\textbf{Acknowledgments.} The work of D.~Bykov was performed at the Steklov International Mathematical Center and supported by the Ministry of Science and Higher Education of the Russian Federation (agreement no. 075-15-2022-265; section 1 of the present paper). Sections 2, 3, 4 of the present paper were written with the support of the Russian Science Foundation grant №~22-72-10122 (\href{https://rscf.ru/en/project/22-72-10122/}{\emph{https://rscf.ru/en/project/22-72-10122/}}). We would also like to thank 
E. Ivanov and A. Nersessian for illuminating discussions, and the Institut des Hautes \'Etudes Scientifiques, where our collaboration emerged, for hospitality.

\vspace{1cm}
\appendix
\section{Multiplicities of degenerate multiplets}
\setcounter{equation}0
\def\theequation{A.\arabic{equation}}
We address the reader to the excellent Slansky's review including extensive group theory information and the tables where the dimensions and other properties of many $SU(n)$ multiplets are listed \cite{Slansky}. Here we only wish to illustrate the derivation of the formula \p{dim-q} for the number of the states \p{Psi-Ans-q}  in a multiplet at the level $L$ in the bosonic $\mathbb{CP}^{n-1}$ model at the presence of a $\mathbb{CP}$ monopole of charge $q \geq 0$. As was mentioned, this number coincides with the dimension of the representation of $SU(n)$ with the highest weight $\Lambda = (L, n-3 \ {\rm zeros}, L+q)$.

The dimension of a representation of with the highest weight $\Lambda$ is given by the Weyl formula,
\be
\lb{Weyl}
\#(\Lambda) \ =\ \prod_{\alpha \in \Delta^+}  \frac {(\Lambda + \delta, \alpha)}{(\delta, \alpha)}\,,
\ee
where the product is done over all positive roots of the algebra and $\delta$ is the highest root.
In our case, 
\be\Lambda = (L, \underbrace{0,\ldots, 0}_{n-3},  L+q) \quad \mbox {\rm and} \quad
 \delta = (\underbrace{1,\ldots,1}_{n-1})\,.
\ee

  \begin{figure}[h]
   \begin{center}
 \includegraphics[width=.40\textwidth]{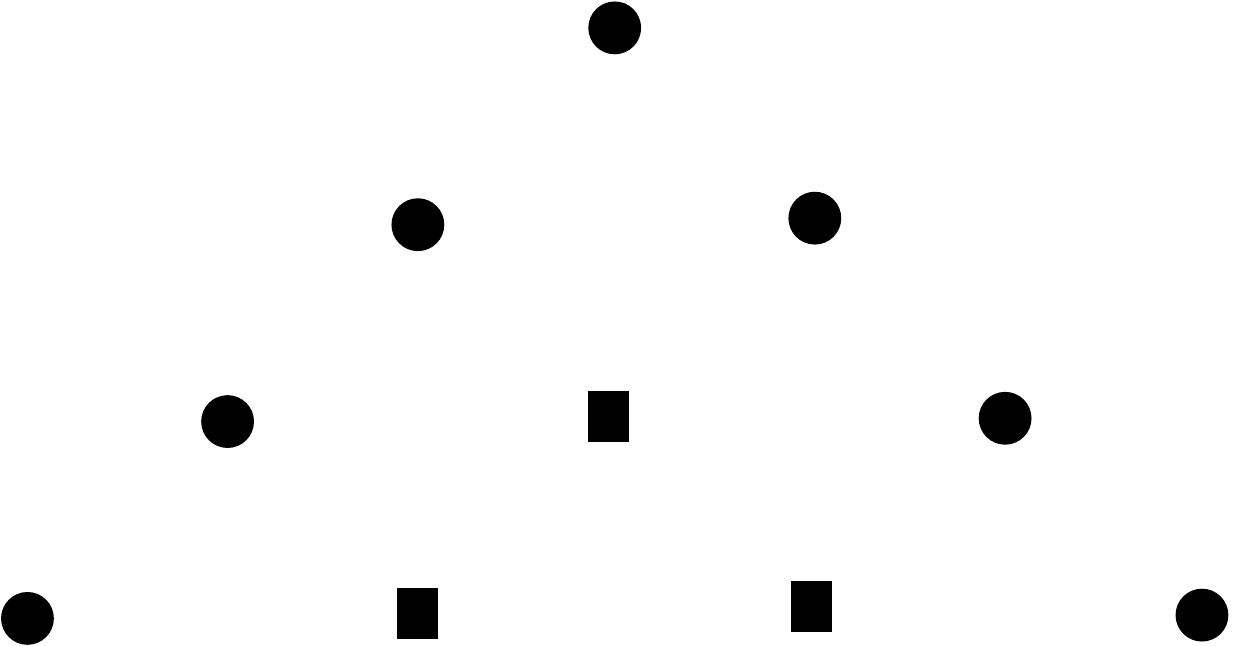}
    \end{center}
\caption{The positive roots of $su(5)$} 
\label{roots}
\end{figure} 
Take for example the case $n=5$. The system of positive roots is drawn in Fig.\ref{roots}. The low side of the triangle represents the simple roots,
\be
 \beta_1 \ =\ (1,0,0,0), \ \beta_2 \ =\ (0,1,0,0), \ \beta_3 \ =\ (0,0,1,0), \ \beta_4 \ =\ (0,0,0,1) \,,
\ee
and the summit of the triangle is the highest root, $\delta =  \beta_1 + \beta_2 + \beta_3 + \beta_4$. The boxes mark the roots $\alpha$ where the inner products $(\Lambda + \delta, \alpha)$ and $(\delta, \alpha)$ coincide, so that these roots do not contribute in the product \p{Weyl}. The circles mark the roots where $(\Lambda + \delta, \alpha) \neq (\delta, \alpha)$ and which give a contribution. We obtain
\be
\#(n=5, L,q)  &&=\ (L+1)\cdot \frac {L+2}2 \cdot  \frac {L+3}3 \cdot   \frac {2L+ q + 4}4 \cdot \frac {L+q+ 3}3 \cdot  \frac {L+q+2}2  \cdot (L+q+1) \ \nn
  && =\
\left( \begin{array}{c} L+3\\ 3 \end{array} \right) \left( \begin{array}{c} L+q + 3\\ 3 \end{array} \right)  \frac {2L+ q + 4}4  \,,
 \ee
which agrees with \p{dim-q}.

A generalization for an arbitrary $n$ is straightforward. Only the roots dwelling on the lateral sides of the triangle of positive roots contribute.   

Note that the same result for the multiplicity could be derived in a more elementary way without invoking group theory wisdom. The  tensor $A_{(\alpha_1 \ldots \alpha_L)|(\bar\beta_1 \ldots \bar\beta_{L+q})}$ would have
 \be
\lb{all}
\left( \begin{array}{c} L+n-1\\ n-1 \end{array} \right)  \left( \begin{array}{c} L+q + n-1\\ n-1 \end{array} \right)
\ee
independent components if no constraints were present. But the constraint of tracelessness imposes 
   \be
\lb{constr}
\left( \begin{array}{c} L+n-2\\ n-1 \end{array} \right)  \left( \begin{array}{c} L+q + n-2\\ n-1 \end{array} \right)
\ee
conditions. The difference of \p{all} and \p{constr} coincides with \p{dim-q}.

\section{Upper components of the superdoublets}\label{updoubletapp}
\setcounter{equation}0
\def\theequation{B.\arabic{equation}}
 The upper components of the supersymmetric doublets are obtained from the lower ones by the action of the supercharge $\hat Q = -i \chi^\alpha \dd/\dd w^\alpha$. Since the supercharges commute with $SU(n)$, the upper components furnish the same representation of $SU(n)$ as the lower ones.
 
As an illustration, let us take $n=4$ and consider the action of $\hat Q$ on a wave function $\Psi_-$ characterized by the parameters $L=2, L' = 1, F=1$ and decribed by the diagram
 \be 
\lb{Young-Psi-}
\ytableausetup{boxsize=0.4cm}
\ytableaushort
{} * [*(lbluecol)]{3,2,1} 
\ee
The corresponding wave function has the form
\be
\lb{Eqn-Psi-}
\Psi_- \ =\    \frac {A_{\alpha_1\alpha_2|\bar\beta|\gamma} \, w^{\alpha_1}w^{\alpha_2} \bar w^\beta \chi^\gamma}{{\cal X}^3}\,.
\ee
For its superpartner, we derive
\be
\Psi_+ =  i \hat Q \Psi_- \ =\ A_{\alpha_1\alpha_2|\bar\beta|\gamma} w^{\alpha_2} \bar w^\beta \chi^\gamma \,\frac {2\bar w^\delta w^\delta \chi^{\alpha_1} - 3 \chi^\delta \bar w^\delta w^{\alpha_1}}{{\cal X}^4}
\ee

This is a wave function \p{PsiFermMon} with $L = L' = F =2$. It satisfies the constraint $\hat {\cal C}_3 \Psi_+ = 0$ and hence, as was explained in the main text, is described by a hook
 \be
\lb{hook}
\ytableausetup{boxsize=0.4cm}
\ytableaushort
{} * [*(bluecol)]{2} * [*(pinkcol)]{1,1,1}
\ee
The irreducible representation to which $\Psi_+$ belongs is given by one of the terms in the decomposition
 \be
\lb{n4-example}
\ytableausetup{boxsize=0.4cm}
\ytableaushort
{}* [*(lbluecol)]{2,2,2} \otimes  \ytableaushort {}* [*(bluecol)]{2} * [*(pinkcol)]{1,1,1} \quad \ = \quad \ 
\ytableaushort {} * [*(lbluecol)]{2,2,2}  * [*(bluecol)]{4} *[*(pinkcol)]{3,3,3} \quad \oplus \quad 
\ytableaushort {} * [*(lbluecol)]{1,1,1}  * [*(bluecol)]{3} *[*(pinkcol)]{2,2}   \quad \oplus \quad 
\ytableaushort {} * [*(lbluecol)]{1,1,1}  * [*(bluecol)]{2} *[*(pinkcol)]{2,2,2}   \quad \oplus \quad 
\ytableaushort {}  * [*(bluecol)]{1} *[*(pinkcol)]{1,1}
\ee
 And we see that the {\it second} term of this decomposition exactly coincides with \p{Young-Psi-}!

It is not the first highest weight diagram, as was the case in Eq.\p{tensprodyoung},  because the  tracelessness conditions \p{albet-trace} and \p{betgam-trace} are now absent.
The second term in \p{n4-example} is obtained when one of the {\it red} ``fermion" boxes in the hook is glued down on the left to the large rectangle so that the left column of the rectangle disappears. Previously, we dropped this contribution due to the condition $\hat {\bar Q} \Psi_- = 0$ and its corollary \p{betgam-trace}. But\footnote{We are not discussing  zero modes here, which may appear in the sector $F = n-1$ and are singlets under supersymmetry.} $\hat{\bar Q}\Psi_+ \neq 0$. 
The described rule of thumb to draw the  Young diagram for the relevant multiplet of $\Psi_+$ --- {\it consider the product of a rectangle associated with $\bar w$ factors  and  a hook associated with $w$ and $\chi$ factors, take a fermion box in the $w\chi$ hook, attach it down on the left to the $\bar w$  rectangle and glue the rest of the hook on the right of the rectangle} --- also works in all other cases.

\vspace{0.3cm} 
Finally, let us prove the fact mentioned in the footnote  on p.~\pageref{albet-trace}.

\begin{thm}
Let 
\be 
\lb{Psiplus}
\Psi_+  = \hat{Q} \Psi_-  = \frac {1}{{\cal X}^M}\, B_{\alpha_1 \ldots \alpha_L| \bar\beta_1 \ldots \bar\beta_{L'}| \gamma_1 \ldots \gamma_F} w^{\alpha_1} \cdots w^{\alpha_L} \, 
 \bar w^{\beta_1} \cdots \bar w^{\beta_{L'}}\,\chi^{\gamma_1}\cdots \chi^{\gamma_F}\  \neq 0\,,
\ee
where $\Psi_-$ is an eigenfunction of the Hamiltonian \p{super-H} satisfying  $\hat {\bar Q} \Psi_- =  0$. Then the tensor~$B$ does not fulfil the $\alpha\beta$ tracelessness condition \p{albet-trace}.
\end{thm}
\begin{proof}
Note first that, for the function defined in \p{Psiplus}, 
\be
\lb{neq-0} 
\hat {\bar Q}\Psi_+  \ \propto  \  \Psi_- \neq 0
\ee
and $F \geq 1$.

Suppose that $B$ is $\alpha\beta$ traceless. Then, in virtue of Theorem 1, \p{Psiplus}  
 would be an eigenfunction of the first term in~\p{super-H} and therefore of the second term as well, meaning that
\be\label{2termeigenprob}
\bar w^\alpha \chi^\alpha  \frac{\partial^2 \Psi_+}{\partial \chi^\beta \partial \bar w^\beta} =a\, \Psi_+\,.
\ee
\begin{lem}
$a$ must be different from zero.
\end{lem}
\begin{proof}
If $a$ is zero, then $\frac{\partial^2 \Psi_+}{\partial \chi^\beta \partial \bar w^\beta}$ must be proportional to  $\bar w^\alpha \chi^\alpha$. To show that it is not possible, we represent $\Psi_+$ as 
$\Psi_+\ =\ \frac {\Phi_+}{{\cal X}^{M}}$ to find
\be
\frac{\partial^2 \Psi_+}{\partial \chi^\beta \partial \bar w^\beta}=-{M\over \cal X}\,i\hat{\cal C}_3 \Psi_++{1\over {\cal X}^{M}} \frac{\partial^2 \Phi_+}{\partial \chi^\beta \partial \bar w^\beta}={1\over {\cal X}^{M}} \frac{\partial^2 \Phi_+}{\partial \chi^\beta \partial \bar w^\beta}\,.
\ee
If the tensor $B$ entering $\Phi_+$ were $\beta\gamma$ traceless, the wave function $\Psi_+$ would be annihilated by the action of $\hat {\bar Q} \propto \frac{\dd^2}{\partial \chi^\beta \partial \bar w^\beta}$ in contradiction with \p{neq-0}. However, we may cast~$\Phi_+$ in the form
\be
\lb{Fi1+Fi2}
\Phi_+ \ = \ B^{(1)}_{\alpha_1 \ldots \alpha_L| \bar\beta_1 \ldots \bar\beta_{L'}| \gamma_1 \ldots \gamma_F} w^{\alpha_1} \cdots w^{\alpha_L} \, 
 \bar w^{\beta_1} \cdots \bar w^{\beta_{L'}}\,\chi^{\gamma_1}\cdots \chi^{\gamma_F} \nn
 + \ \left(\bar w^\alpha \chi^\alpha \right)\,B^{(2)}_{\alpha_1 \ldots \alpha_L| \bar\beta_1 \ldots \bar\beta_{L'-1}| \gamma_1 \ldots \gamma_{F-1}} w^{\alpha_1} \cdots w^{\alpha_L} \, 
 \bar w^{\beta_1} \cdots \bar w^{\beta_{L'-1}}\,\chi^{\gamma_1}\cdots \chi^{\gamma_{F-1}}  \nn
  \equiv\ \Phi_+^{(1)}+ \left(\bar w^\alpha \chi^\alpha \right)\,\Phi_+^{(2)}\,,
\ee
where both $B^{(1)}$ and $B^{(2)}$ are $\beta\gamma$ traceless. 

We act on \p{Fi1+Fi2} with the operator $\frac{\dd^2}{\partial \chi^\alpha \bar w^\alpha}$. The first term does not contribute due to the  $\beta\gamma$ tracelessness of $B^{(1)}$, and it is easy to derive 
\be
\frac{\partial^2 \Phi_+}{\partial \chi^\beta \partial \bar w^\beta}\ =\ \left(n+L'-F\right)\,\Phi_+^{(2)}\,.
\ee
Since $\Phi_+^{(2)}$ does not contain a factor of  $\bar w^\alpha \chi^\alpha$, one cannot have $\frac{\partial^2 \Psi_+}{\partial \chi^\beta \partial \bar w^\beta} \propto 
\bar w^\alpha \chi^\alpha 
$. Hence~${a\neq 0}$.
\end{proof}

 It then follows from~(\ref{2termeigenprob}) that $\Psi_+$ is proportional to $\bar w^\alpha \chi^\alpha$, i.e.
\be\label{Psiplusform}
\Psi_+\ =\ \frac {\Phi_+^{(2)}}{{\cal X}^{M}}\,\left(\bar w^\alpha \chi^\alpha \right)\,,
\ee
and $\Phi^{(1)}_+ = 0$.

On the other hand, $\Psi_+$ is obtained by acting with $\hat Q$ on a state $\Psi_-$ with a  wave function~\p{PsiFermMon} including a tensor\footnote{Or just $A_{\alpha_1 \ldots \alpha_{L+1}| \bar\beta_1 \ldots \bar\beta_{L'}}$ if $F=1$.} 
$$A_{\alpha_1 \ldots \alpha_{L+1}| \bar\beta_1 \ldots \bar\beta_{L'}| \gamma_1 \ldots \gamma_{F-1}}  $$
satisfying both the $\alpha\beta$ and $\beta\gamma$ tracelessness conditions \p{albet-trace}, \p{betgam-trace}
As a result of such action, one can get the factor $\left(\bar w^\alpha \chi^\alpha \right)$, as in the formula above, only when the derivative ${\dd \over \dd w^\alpha}$ acts on the $\cal X$-factors in the denominator of~$\Psi_-$. Indeed, when the derivative acts on the numerator, one obtains a structure including the tensor
\be  
B_{\alpha_1 \ldots \alpha_L| \bar\beta_1 \ldots \bar\beta_{L'}| \gamma_1 \ldots \gamma_{F-1}\alpha_{L+1}}  \ =\ 
A_{\alpha_1 \ldots \alpha_{L+1}| \bar\beta_1 \ldots \bar\beta_{L'}| \gamma_1 \ldots \gamma_{F-1}}\,.
\ee
This tensor is $\beta\gamma$ traceless due to the $\alpha\beta$ and $\beta\gamma$ tracelessness of $A$ and gives a nonzero contribution to $\Phi_+^{(1)}$. We are thus led to a contradiction.

\end{proof}

\end{document}